\documentclass[fleqn]{elsarticle}
\usepackage[mathscr]{eucal}
\usepackage{amsmath,amsfonts,amssymb,amsthm}
\usepackage{graphicx}
\usepackage{mathtools}
\usepackage[dvipsnames]{xcolor}
\usepackage{booktabs}
\usepackage{tabularx}
\usepackage{hyperref}
\usepackage{multirow}
\usepackage[subnum]{cases}
\usepackage{color}
\usepackage{cases}
\usepackage{natbib}
\usepackage{tgpagella}
\usepackage{subfigure}
\usepackage{caption}
\usepackage[english]{babel}
\theoremstyle{plain}
\newtheorem{theorem}{Theorem}
\newtheorem{conjecture}{Conjecture}

\newtheorem{lemma}{Lemma}

\newtheorem{proposition}{Proposition}

\newtheoremstyle{note}{\topsep}{\topsep}{\slshape}{}{\scshape}{}{ }{}
\theoremstyle{note}

\biboptions{sort&compress}
\numberwithin{equation}{section}
\numberwithin{definition}{section}

\mathtoolsset{mathic,centercolon}
\newcommand\bq{{\mathbf q}}

\newcommand\bv{{\mathbf v}}

\newcommand\bF{{\mathbf F}}

\newcommand\bM{{\mathbf M}}

\newcommand\scE{{\mathscr E}}

\newcommand\scG{{\mathscr G}}

\newcommand\scN{{\mathscr N}}

\newcommand\mvector{\boldsymbol}

\newcommand\vp{\mvector{p}}
\newcommand\vq{\mvector{q}}

\newcommand\vv{\mvector{v}}

\newcommand\vx{\mvector{x}}
\newcommand\vy{\mvector{y}}
\newcommand\vz{\mvector{z}}
\newcommand\vA{\mvector{A}}

\newcommand\vX{\mvector{X}}

\newcommand\vvarphi{\mvector{\varphi}}

\newcommand\field{\mathbb}

\newcommand\R{\field{R}}

\newcommand\C{\field{C}}
\newcommand\Z{\field{Z}}

\newcommand\N{\field{N}}
\newcommand\Q{\field{Q}}

\newcommand\const{\operatorname{const}}

\newcommand\rmd{\mathrm{d}}

\newcommand\rmi{\mathrm{i}\mspace{1mu}}
\newcommand\rme{\mathrm{e}}

\newcommand\Dz{\frac{\mathrm{d}\phantom{z} }{ \mathrm{d}z}}

\usepackage{subfigure}
\begin{document}
	
\begin{frontmatter}
	\title{Dynamics and non-integrability of the variable-length double pendulum: exploring chaos and periodicity via the Lyapunov refined maps}
	
	\author{Wojciech Szumi\'nski$^{1*}$}
	\ead{w.szuminski@if.uz.zgora.pl}
	\address{$^1$Institute of Physics, University of Zielona G\'ora, Licealna 9, PL-65-407, Zielona G\'ora, Poland}
		\author{Tomasz Kapitaniak$^{2}$}
	\ead{tomasz.kapitaniak@p.lodz.pl}
	\address{$^2$Division of Dynamics, Lodz University of Technology, Stefanowskiego 1/15, 90–537 Lodz,
Poland}
	\begin{abstract}
This paper extends our previous work~(Szumi\'nski and Maciejewski, 2024), where we explored the dynamics and integrability of the double-spring pendulum. Here, we investigate the variable-length double pendulum, a three-degree-of-freedom Hamiltonian system combining features of the classic double pendulum and the swinging Atwood machine. With its intricate dynamics, this system is crucial for studying nonlinear phenomena such as high-order resonances, chaos, and bifurcations. We address the challenges posed by high-dimensional phase spaces using a novel tool, the \textit{Lyapunov refined maps}, which integrates Poincar\'e sections, phase-parametric diagrams, and Lyapunov exponents. This framework comprehensively analyzes periodic, quasi-periodic, and chaotic behaviors. By measuring the strength of chaos, it also offers insights into the system's dynamical structure. Additionally, we apply Morales-Ramis theory to examine integrability, leveraging the differential Galois group of variational equations to establish non-integrability conditions. The Kovacic algorithm is used to analyze the solvability of higher-dimensional differential equations, complemented by Lyapunov exponent diagrams to exclude integrable dynamics under certain parameters. Our findings advance the fundamental understanding of variable-length pendulum dynamics, offering new insights and methodologies for further research with potential applications in adaptive robotics, energy harvesting, and biomechanics. Additionally, this work represents a significant step toward proving the long-sought non-integrability of the classical double pendulum.
\end{abstract}
	
	\begin{keyword}
 Variable-length pendulum; Ordinary differential equations; Chaos in Hamiltonian
 systems;   Integrability  \paragraph{Declaration} The article has been published in~\cite{Szuminski:25::}, and the final version is available at: \textbf{\href{https://doi.org/10.1016/j.jsv.2025.119099}{https://doi.org/10.1016/j.jsv.2025.119099} }
\end{keyword}

\end{frontmatter}
\section{Introduction  and motivation}
The study of pendulum systems represents one of the cornerstones of classical mechanics and dynamical systems, dating back to the foundational work of Galileo Galilei. The pendulum’s simple structure - a mass suspended from a fixed point and oscillating under the influence of gravity — presents a deceptively simple system that exhibits regular motion. However, despite such a simple form of the equation of motion $\ddot \varphi+\sin \varphi=0$, its solutions are given in terms of the Jacobi elliptic functions. Thus, the possibility of finding an explicit form of solutions to systems of differential equations is not obvious and - in fact - is something special. Therefore, the concept of integrability of a system was introduced, which is easier to analyze
than its solvability. Integrability is related to the existence of first integrals (conservation laws) or other invariant tensor fields, which remain constant on all solutions of equations of motion describing the studied system. 

Historically, pendulum systems have played a key role in understanding how regular motion can shift to chaotic behavior.  The paradigm models in pendulum dynamics include the spring pendulum ~\cite{Broucke:73::,Lee:97::,Maciejewski:04::c}, the magnetic pendulum~\cite{Moon:87::,Kecik:15::,Kwuimy:12::,Chen:22::}, and multi-body systems such as the double and triple pendulums~\cite{Shinbrot:92::,Stachowiak:06::,mp:13::c,Stachowiak:15::,PUZYROV2022116699,Nigmatullin:14::,MR4191961,MR4412899,Puzyrov:22::},  the coupled pendulums~\cite{Huynh2010,Huynh2013,Elmandouh:16::,Szuminski:20::,Szuminski:23::} and the swinging Atwood machine~\cite{Tufillaro:84::,Tufillaro:90::,Szuminski:22::,Olejnik:23b::}. Such systems serve as essential models in studying chaotic motion, bifurcations, and stability transitions. Researchers use them as benchmarks to understand how small changes in initial conditions can lead to completely different behaviors.  Indeed, these models have been extensively studied by researchers both theoretically and practically~\cite{Levien:93::,Pujol:10::,Gomez:21::,ciezkowski:21::,Pilipchuk:22::,Chu:22::,Olejnik:23b::}.

Currently, the various types of variable-length pendulums are undergoing extensive analysis across different methods and techniques~\cite{Yakubu:22::,Szuminski:22::,Szuminski:23::,Olejnik:23b::,Szuminski:24::,Yakubu:24::}.  Due to their complex dynamics and sensitivity to initial conditions, they have extensive physical, astronomical,  and engineering applications. 
 For instance, the spring pendulum can be treated as a classical analog of the quantum phenomenon of 	 Fermi resonance in the infrared spectrum of carbon dioxide~\cite{Vitt:33::,MR1751314}, and currently, it has been treated as a system with potential applications in the atmosphere modeling~\cite{MR1948160,MR2043791,DeShazer}. Whereas the double spring pendulum model as studied in~\cite{Szuminski:24::}, can be treated as a hypothetical model of the active debris removal
missions~\cite{Ledkov:19::,Shi:18::,Aslanov:24::,Bourabah:22::}.   In engineering, however,  the variable-length pendulums have 
 many applications in harvesters, load-lifting equipment, and robotics~\cite{Wojna:18::,Mahmoudkhani:18::,Liu:19::,Sharghi:22::,Yang:22::,Pilipchuk:22::,Chu:22::,Olejnik:23b::,Meng:24::}.
 Analyzing the dynamics of the double pendulum system helps engineers design control algorithms, especially for bipedal robots or in aerodynamics where motion stability is a critical aspect~\cite{Selyutskiy:19::,4153387,Sahin:17::}.

%The dynamics of the pendulum systems have been extensively studied across a wide array of scientific and engineering disciplines, resulting in a wealth of literature focused on their chaotic behavior, stability, and sensitivity to initial conditions~\cite{Broucke:73::,Lee:97::,Maciejewski:04::c,ZHANG2020115549,Elmandouh:16::,Szuminski:20::}.
 As one of the most accessible and illustrative examples of chaotic systems, the classic double pendulum is a cornerstone model for studying nonlinear dynamics and chaotic motion~\cite{Dullin:94::,Ivanov:99::}. Numerous publications have explored its dynamics from various perspectives, including analytical studies on its stability and integrability, computational simulations to capture its chaotic motion~\cite{Yu:98::,Shinbrot:92::}, and experimental setups to observe its motion in real time~\cite{Wojna:18::}. These studies have provided a deep understanding of the fundamental principles of chaos theory and nonlinear mechanics, highlighting the double pendulum’s role as a pivotal system in dynamical systems research. Moreover, plenty of modifications of the original double pendulum have been studied. For instance, the restricted double pendulum~\cite{Samaranayake:93::,Stachowiak:15::}, the magnetic double pendulum~\cite{Wojna:18::,Polczynski:19::,Skurativskyi:22::,Moatimid:23::}, double pendulum with distributed mass~\cite{Rafat:09::,Huang:15::}, Rott's pendulum~\cite{Dyk:24::,Dyk:24b::} or the  rotating double pendulum~\cite{Maiti:16::,Kudra:19::}, to mention just a few.

Despite the extensive research on the fixed-length double pendulum, there is a clear gap in the study of the variable-length double pendulum. This version adds complexity because it allows one or both arms to change length, creating adjustable parameters that affect its dynamic behavior. Unlike the fixed-length system, the variable-length double pendulum presents unique challenges for both theoretical analysis and practical experimentation due to its higher-dimensional phase space and the increased number of variables that affect its motion. This complexity not only makes the system more difficult to analyze using standard techniques but also opens up new avenues for research, particularly in understanding how length variability influences the onset of chaos, stability, and resonance phenomena.
Thus, while the studies of the classic double pendulum are well established, the variable-length double pendulum remains relatively unexplored. This lack of comprehensive analysis points to an opportunity for new studies that could bridge the gap in our understanding, offering insights into the dynamic behavior of systems with adaptable structural configurations. By addressing this gap, future research could reveal novel properties and potential applications of variable-length pendulum systems in fields such as adaptive robotics, energy harvesting, and biomechanics, where variable geometries play a crucial role in system performance and functionality~\cite{PLAUT20133768,YANG2022116727,SHARGHI2022117036,Olejnik:23b::}. For instance, variable-length pendulum systems serve as excellent models for cranes and lifting equipment, where understanding motion and stability is essential for safe and efficient operation~\cite{JU2006376,MR4459645,Freundlich:20::,Shahbazi:16::}. Additionally, combining variable-length double pendulum systems with the swinging Atwood machine may have applications in energy conversion and storage, where the swinging motion can be used to generate electricity~\cite{MARSZAL2017251,He:22::,ABOHAMER2023377}.

In this paper, we study the dynamics and integrability of a variable-length double pendulum system that integrates aspects of the double pendulum with the swinging Atwood machine. This three-degree-of-freedom system serves as a natural extension of the classical double pendulum model but introduces additional complexity due to the interplay between the pendular motion and the variable-length constraint. This coupling creates a richer dynamical landscape, one that not only challenges existing numerical methods and techniques but also provides a unique opportunity to extend methodologies in the study of non-integrable systems. Indeed, the best tool to get a quick insight into the dynamics of Hamiltonian systems of two degrees of freedom is the Poincar\'e cross-section method. It provides a beautiful coexistence of periodic, quasi-periodic, and chaotic motion at the cross-section plane, giving a qualitative insight into system dynamics.  
However, the phase space is six-dimensional in a three-degree-of-freedom Hamiltonian system, such as the proposed variable-length double pendulum. The Poincar\'e section in this scenario would reduce the dimensionality by just one, resulting in a five-dimensional hypersurface defined by a constant energy level, which is beyond our ability to visualize or interpret in a physically intuitive way directly. Indeed,  resonances and quasi-periodic orbits make structures more complex than simple closed curves or tori.
Consequently, much of the valuable qualitative insight is lost, making the Poincar\'e sections method almost unusable in this case.  For these reasons, alternative approaches are needed.

On the other side, even though the Lyapunov exponents provide a quantitative description of chaos and can be effectively applied to a Hamiltonian system with many degrees of freedom, it says nothing about the existence of periodic solutions, their numbers or periodicity in a regular regime, where all exponents are zero. This ambiguity poses a considerable challenge, as understanding the existence of resonance orbits is crucial in Hamiltonian dynamical systems. However, constructing phase-parametric diagrams provides an effective approach to address this issue.  In this paper, we continue and extend our previous works as evidenced in~\cite{Szuminski:23::,Szuminski:24::}, by
combining three numerical methods such as the Poicnar\'e sections, the phase-parametric diagrams, and the Lyapunov exponents, as one powerful tool, named the \textit{Lyapunov refined maps}.
This approach provides a complete view of the system’s dynamics, capturing both qualitative and quantitative aspects by assessing the strength of chaos and covering periodic, quasi-periodic, and chaotic behaviors.

However, due to technical limitations, even the best numerical analysis can only partially determine whether a given dynamical system that depends on several parameters is integrable for certain values of these parameters.
To find new integrable cases or prove the considered model's nonintegrability, one can use the Morales–Ramis theory~\cite{Morales:99::,Morales:00::}. 
 Morales-Ruiz and Ramis showed that integrability in the Liouville sense imposes a very restrictive condition for the identity component of the differential Galois group of variational equations obtained by the lineralization of the equations of motion along a particular non-equilibrium solution. The main theorem of this theory states. 
\begin{theorem}[Morales--Ramis (1999)]
 If a Hamiltonian system is integrable in the sense of Liouville in a neighborhood of a particular solution, then the identity component of the
  differential Galois group of the variational equations along this solution is Abelian.
\end{theorem}	
	Over the past twenty years, Morales–Ramis theory has been successfully applied to a wide range of important physical systems~\cite{Yagasaki:18::,Acosta:18::,Acosta:18b::,Huang:18::,Combot:18::,Mnasri:18::,Shibayama:18::,Maciejewski:18::}, including non-Hamiltonian systems~\cite{Huang:18::,Szuminski:18::,Maciejewski:20e::,Szuminski:20c::}, to name just a few. This has led to the discovery of numerous new integrable and super-integrable systems~\cite{Maciejewski:05::b,Elmandouh:18::,Szuminski:18a::,Szuminski:18b::}. 

The main idea of the above theory is rather simple; we compute the variational equations along a particular non-equilibrium solution, and we analyze the differential Galois group of the variational system. However, this analysis is the most challenging point, since there is no universal method to check whether the identity component of the differential Galois group is Abelian or not. For Hamiltonian systems of two degrees of freedom,  for which in most cases the variational equations can be transformed into the one second-order rational, differential equation, 
there exists an algorithm called the Kovacic algorithm~\cite{Kovacic:86::}, which classifies the possible types of solutions of such equation. As a by-product,  it determines the differential Galois group of the equation.  In Hamiltonian systems with three degrees of freedom, the normal variational equations typically form a four-dimensional subsystem, making the analysis of its differential Galois group considerably more complex. Unfortunately, no equivalent of the Kovacic algorithm exists for higher-order linear differential equations with rational coefficients, although some partial results are available~\cite{Singer:95::,Ulmer:03::}. Recent work~\cite{Combot:18b::} introduces an algorithm for analyzing the differential Galois group of symplectic differential operators in four dimensions, representing perhaps the most comprehensive study to date. 

At the end of this introduction, it is worth noting that for the effective application of Morales–Ramis theory and the whole differential Galois group framework, it is necessary to identify a non-trivial particular solution to the equations of motion for the given dynamical system. This requirement is precisely why a non-integrability proof for the classical double pendulum is still missing — an explicit particular solution has yet to be found. However, in the proposed model, we can obtain a particular solution, and with the help of the four-dimensional Kovacic algorithm, we can establish the non-integrability of the variable-length double pendulum. This result brings us closer than ever to proving the non-integrability of the classical double pendulum. 

The structure of this paper is as follows. In Sec.~\ref{sec:model_1}, we present the model and derive the associated equations of motion. Sec.~\ref{sec:numerical} provides a detailed analysis of the system’s chaotic behavior through qualitative and quantitative approaches, employing numerical techniques such as Lyapunov exponent spectra, phase-parametric diagrams, and Poincar\'e sections. We combine these three numerical methods to build the Lyapunov refined maps.
 In Sec.~\ref{sec:integrability}, we conduct a rigorous integrability analysis using Morales–Ramis theory alongside the four-dimensional Kovacic algorithm, ultimately proving the system’s non-integrability. Essential concepts and theorems relevant to the integrability study of four-dimensional variational equations are provided in the Appendix.
		\section{Description of the system\label{sec:model_1}}
		
	% TODO: \usepackage{graphicx} required
 We consider the double variable-length pendulum with counterweight mass, recently studied in~\cite {Yakubu:22::,Olejnik:23b::}.  As depicted in Fig.~\ref{fig:geometria}, 
the system consists of three masses: $M$, $m_1$, and $m_2$. The masses $M$ and $m_1$ are linked by a non-stretchable string of constant length $d$. The mass $m_1$ is allowed to swing in a
plane, while the second mass $M$  
acts as a counterweight and can only move vertically. Therefore, the masses
$M,m_1$ form a system known as the swinging Atwood machine. However,
to analyze the double pendulum with a variable length, we attach to the mass $m_1$, the second with mass $m_2$ and length $l_2$.  
		\begin{figure}[t]
\centering{
\resizebox{75mm}{!}{\begingroup%
		  \makeatletter%
		\providecommand\color[2][]{%
			\errmessage{(Inkscape) Color is used for the text in Inkscape, but the package 'color.sty' is not loaded}%
			\renewcommand\color[2][]{}%
		}%
		\providecommand\transparent[1]{%
			\errmessage{(Inkscape) Transparency is used (non-zero) for the text in Inkscape, but the package 'transparent.sty' is not loaded}%
			\renewcommand\transparent[1]{}%
		}%
		\providecommand\rotatebox[2]{#2}%
		\newcommand*\fsize{\dimexpr\f@size pt\relax}%
		\newcommand*\lineheight[1]{\fontsize{\fsize}{#1\fsize}\selectfont}%
		\ifx\svgwidth\undefined%
		\setlength{\unitlength}{248.83760708bp}%
		\ifx\svgscale\undefined%
		\relax%
		\else%
		\setlength{\unitlength}{\unitlength * \real{\svgscale}}%
		\fi%
		\else%
		\setlength{\unitlength}{\svgwidth}%
		\fi%
		\global\let\svgwidth\undefined%
		\global\let\svgscale\undefined%
		\makeatother%
 \begin{picture}(1,0.56950485)%
    \put(0,0){\includegraphics[width=\unitlength,page=1]{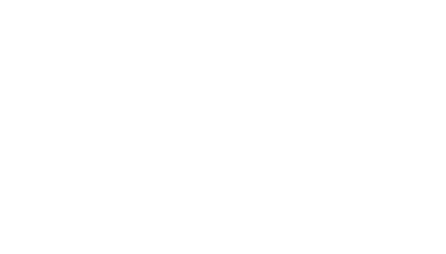}}%
    \put(0.10089368,0.17198639){\color[rgb]{0,0,0}\makebox(0,0)[lt]{\smash{$m_1$}}}%
    \put(0,0){\includegraphics[width=\unitlength,page=2]{geom_2.pdf}}%
    \put(0.48028677,0.01792731){\color[rgb]{0,0,0}\makebox(0,0)[lt]{\smash{ $m_2$}}}%
    \put(0.14140045,0.20174878){\color[rgb]{0,0,0}\makebox(0,0)[lt]{\smash{ }}}%
    \put(0.05040873,0.2916825){\color[rgb]{0,0,0}\makebox(0,0)[lt]{\smash{ $\varphi_1$}}}%
    \put(0.19401878,0.12949936){\color[rgb]{0,0,0}\makebox(0,0)[lt]{\smash{ $\varphi_2$}}}%
    \put(0.15099237,0.28805061){\color[rgb]{0,0,0}\makebox(0,0)[lt]{\smash{ $l_1(t)$}}}%
    \put(0.00094453,0.35482695){\color[rgb]{0,0,0}\makebox(0,0)[lt]{\smash{ $x$}}}%
    \put(0.07608132,0.42578338){\color[rgb]{0,0,0}\makebox(0,0)[lt]{\smash{ $y$}}}%
    \put(0.32862837,0.1352914){\color[rgb]{0,0,0}\makebox(0,0)[lt]{\smash{ $l_2$}}}%
    \put(0,0){\includegraphics[width=\unitlength,page=3]{geom_2.pdf}}%
    \put(0.90104761,0.25034449){\color[rgb]{0,0,0}\makebox(0,0)[lt]{\smash{ $X$}}}%
    \put(0.91497539,0.05867106){\color[rgb]{0,0,0}\makebox(0,0)[lt]{\smash{ $M$}}}%
    \put(0,0){\includegraphics[width=\unitlength,page=4]{geom_2.pdf}}%
  \end{picture}%
    	\endgroup}
\caption{The geometry of a variable-length double pendulum with a counterweight mass. }
		\label{fig:geometria}
}
\end{figure}
	Looking at Fig.~\ref{fig:geometria}, we can immediately derive the corresponding kinetic and  potential energies
	\begin{equation}
		\begin{split}
			\label{eq:lagrangian_cartesian}
			\begin{split}
 T=\frac{M}{2}\dot X^2+\frac{m_1}{2}\left(\dot x_1^2+\dot y_1^2\right)+\frac{m_2}{2}\left(\dot x_2^2+\dot y_2^2\right),\quad  V=-g(M X+m_1 x_2+m_2 x_2).
			\end{split}
		\end{split}
	\end{equation}
The system is constrained holonomically, such that the length of the string is constant
	\begin{equation*}
		\sqrt{x_1^2+y_1^2}+X=d=\operatorname{const}.
	\end{equation*}
	Thus, we introduce the polar coordinates according to the constraint
	\begin{equation*}
		\begin{aligned}
			&	x_1=l_1(t)\cos\varphi_1(t),&& y_1=l_1(t)\sin\varphi_1(t),\\
			& x_2=x_1+l_2 \cos\varphi_2(t),&& y_2=y_1+l_2\sin\varphi_2(t).
		\end{aligned}
	\end{equation*}
	In these coordinates, the kinetic and potential energies~\eqref{eq:lagrangian_cartesian}   takes the form
	\begin{equation}
		\begin{aligned}
			\label{eq:lagrangian_polar}
			T&=\left(\frac{M+m_1+m_2	}{2}\right)\dot l_1^2+\left(\frac{m_1+m_2}{2}\right)l_1^2\dot\varphi_1^2+\frac{m_2}{2}l_2^2\dot \varphi_2^2
			+ m_2l_2\dot \varphi_2[l_1\dot\varphi_1\cos(\varphi_1-\varphi_2)+\dot l_1 \sin(\varphi_1-\varphi_2)],\\	V&= g[Ml_1-(m_1+m_2)l_1\cos\varphi_1-m_2 l_2\cos\varphi_2].	\end{aligned}
	\end{equation}
		The Lagrange’s equations of the second kind are used in the following form
\begin{equation}
\label{eq:lag}
\frac{\rmd}{\rmd t}\left(\frac{\partial T}{\partial \dot q_i}\right)-\frac{\partial T}{\partial  q_i}+\frac{\partial V}{\partial \dot q_i}+\frac{\partial R}{\partial \dot q_i}=Q_i,\qquad i=1,2,3,
\end{equation}
where $ R$ stands for the Rayleigh dissipation function and $Q_i$ is the external force acting on the system.  In this context, we are assuming a conservative case, which means that both the Rayleigh dissipation and the external force acting on the system are equal to zero.  Thus, the system's motion is constrained by the energy first integral $E=T+V=\const$.

To minimize the number of parameters and thus simplify our calculations as much as possible, we rescale the length as $\ell(t) = l_1(t)/l_2$ and introduce a new time variable $t \to \omega_0^{-1} t$, where $\omega_0 = \sqrt{g/l_2}$. The total energy $E$ in its dimensionless form now reads
\begin{equation}
		\begin{split}
		\label{eq:energy}
			&E=\frac{1+\mu_1}{2}v^2+\frac{\ell^2}{2}\omega_1^2+\frac{\mu_2}{2}\omega_2^2+\mu_2\sin(\varphi_1-\varphi_2)v\,\omega_2	
			+\mu_2\ell\cos(\varphi_1-\varphi_2)\omega_1\omega_2
			+	\ell(\mu_1-\cos\varphi_1)-\mu_2\cos\varphi_2.
		\end{split}
	\end{equation}
	Here, $		\mu_1\geq 1,$ and $0\leq\mu_2\leq 1$ are dimensionless parameters defined by
	\begin{equation}
		\label{eq:parki1}
		\mu_1\equiv\frac{M}{m_1+m_2},\qquad \mu_2\equiv\frac{m_2}{m_1+m_2}.
	\end{equation}
Although the system is Hamiltonian, we will not introduce canonical variables
because they lead to a much more complicated form of equations of motion.   Denoting $\bq=[\ell,\varphi_1,\varphi_2]^T$, the Lagrange equations of motion~\eqref{eq:lag} can be written in the following matrix form
\begin{equation}
\label{eq:lag_2}
\bM(\bq)\ddot \bq+\bF_c(\bq,\dot \bq)+\bF_q(\bq)=0,
\end{equation}
where the position-dependent mass matrix $\bM$ is as follows
\begin{equation*}
\label{eq:MM}
\bM=\begin{bmatrix}
1+\mu_1&0&\mu_2\sin \Delta_\varphi\\
0&\ell^2&\mu_2\ell \cos\Delta_\varphi\\
\mu_2\sin \Delta_\varphi&\mu_2\ell \cos \Delta_\varphi &\mu_2
\end{bmatrix},  \qquad  \Delta_\varphi\equiv\varphi_1-\varphi_2.
\end{equation*} 
The  vector  of nonlinear coupling terms  $\bF_c$ and the vector of the potential gravitational force $\bF_g$, are given by
\begin{equation*}
\label{eq:Fc_Fg}
\begin{split}
\bF_c=
\begin{bmatrix}
-\ell^2\dot\varphi_1^2-\mu_2\dot\varphi_2^2\cos\Delta_\varphi\\
\mu_2\ell\dot\varphi_2^2\sin\Delta_\varphi+2\ell \dot \ell \dot \varphi_1\\
\mu_2\dot \varphi_1(2\dot \ell\cos\Delta_\varphi -\ell \dot \varphi\sin \Delta_\varphi )
\end{bmatrix},\qquad  \bF_q=\begin{bmatrix}
\mu_1-\cos\varphi_1\\
\ell \sin\varphi_1\\
\mu_2 \sin \varphi_2
\end{bmatrix}.
\end{split}
\end{equation*}
If, we assume that the matrix $\bM$ is non-singular, i.e., $\mu_2\neq 0$, then by introducing the velocity vector $\dot \bq=\bv=[v,\omega_1,\omega_2]^T$, we can rewrite the Lagrange equations~\eqref{eq:lag_2} as a system of six first-order ODEs, which yields
\begin{equation}
\label{eq:rhs}
\dot \bq=\bv,\qquad \dot \bv=-\bM(\bq)^{-1}\left[\bF_c(\bq,\dot \bq)+\bF_q(\bq)\right].
\end{equation}
For $\mu_2=0$, the system reduces to the classical swinging Atwood machine, and therefore, we exclude this value from our further considerations.

\section{The dynamics\label{sec:numerical}}
The initial results of the analysis of the pendulum model obtained in the studies~\cite {Yakubu:22::,Olejnik:23b::}, indicate the highly complex dynamics of the variable-length double pendulums. On the other hand, we notice that for $\mu_2=0$, which implies $m_2=0$, the system and the energy function~\eqref{eq:energy} reduce to the ordinary swinging Atwood machine, which is a system of two degrees of freedom~\cite{Tufillaro:85::,Tufillaro:90::}. Therefore, we start the numerical analysis by assuming small values of $\mu_2$. 
We conduct a numerical analysis to explore the system’s dynamics using three key techniques: Lyapunov exponents diagrams (Lyapunov diagrams, in short), phase-parametric diagrams, and the Poincar\'e cross-sections method. We combine these three numerical methods as one powerful tool to get an exhaustive view of the system's dynamics.

	\subsection{Lyapunov exponents diagrams}
	\begin{figure*}[t]	\centering
\subfigure[$\varphi_{10}=0.25\pi$]{
		\includegraphics[width=0.32\linewidth]{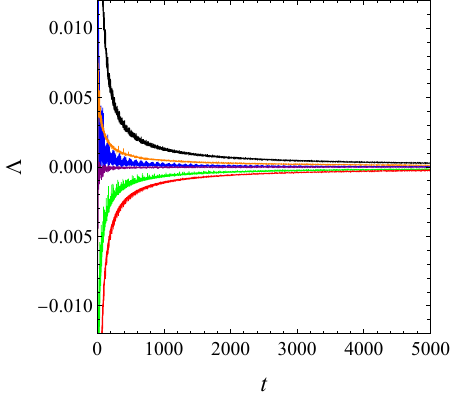}	}
		\subfigure[$\varphi_{10}=0.5\pi$]{\includegraphics[width=0.32\linewidth]{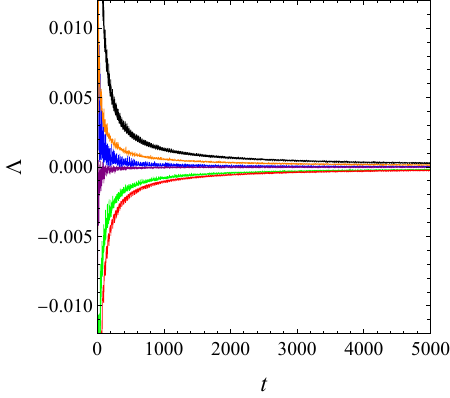}}\subfigure[$\varphi_{10}=0.75\pi$]{
			\includegraphics[width=0.32\linewidth]{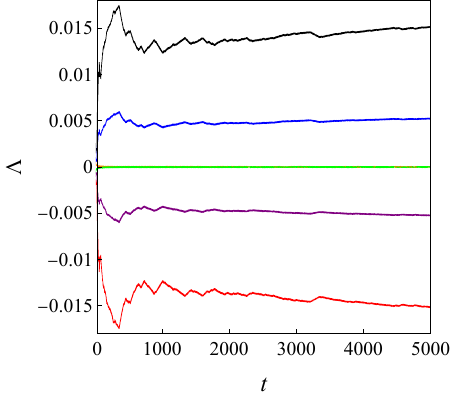}}
		\caption{
(Color online) The Lyapunov exponents spectra $\Lambda=\{\lambda_1,\ldots, \lambda_6\}$ of system~\eqref{eq:rhs}, computed for the constant values of parameters $\mu_1=3, \mu_2=0.2$, and initial condition~\eqref{eq:ini}. 
For a sufficient amount of time steps, the convergences of the Lyapunov exponents are ensured. 		\label{fig:parki0}}
			\end{figure*}

The technique of Lyapunov exponents is perhaps the most fundamental and well-known chaos indicator. It quantifies how rapidly nearby trajectories in phase space diverge exponentially over time. In an $n$-dimensional system of first-order differential equations, $n$ Lyapunov exponents correspond to each principal direction of motion in the phase space. The full set of exponents, known as the Lyapunov spectrum, offers a detailed picture of the system dynamics.  In chaos theory, the onset of chaos is indicated by the presence of one positive Lyapunov exponent. In contrast, hyperchaos is characterized by two or more positive Lyapunov exponents.
 
   Over the years, various algorithms for calculating the Lyapunov spectrum of continuous dynamical systems have been thoroughly developed~\cite{Wolf:85::,Rangarajan:98::,Carbonell:02::,Lu:05::,Chen:06::,Stachowiak:11::,Balcerzak:18::,Balcerzak:20::,Balcerzak:20b::}. Most of these methods build upon the widely recognized standard algorithm proposed by Benettin et al.~\cite{Benettin:80::}.  This technique is based on the integration of variational equations
for $n$ initial conditions periodic application of the Gram–Schmidt orthonormalization procedure.   The algorithm was straightforwardly implemented in  Mathematica by Sandri~\cite{Sandri:96::}. In our research, we use the standard algorithm by utilizing a solver with the fourth-order Runge--Kutta method\footnote{To be precise, we used the built-in Mathematica  NDSolve solver.  We specified the solving method as "ExplicitRungeKutta" by setting "StiffnessTest" $\rightarrow$ True. We set the PrecisionGoal and AccuracyGoal to at least 11 digits while keeping the MaxStepSize as  0.01.}. 
We have chosen key parameters for our computations following an extensive and careful analysis. Specifically, we set the re-orthonormalization interval to $T=1$, a maximum step size of $\tau=0.01$, and typically run the integration for at least $k=5000$ steps.  The working precision for the entire numerical analysis is set to at least 12,
ensuring precision maintenance of up to 12 digits during internal computations. Additionally, we utilize the constancy of the energy first integral $E$ (as defined in~\eqref{eq:energy})  to verify the accuracy of our numerical integrations. The relative and absolute errors are kept within a tolerance of $10^{-11}$, ensuring high accuracy throughout the process.
These parameters enable the efficient computation of Lyapunov exponents, promoting their rapid convergence. We can confidently calculate the Lyapunov spectrum by maintaining high precision and small error margins, even for systems with chaotic dynamics.

Fig.~\ref{fig:parki0} shows three spectrums of Lyapunov exponents for system~\eqref{eq:rhs}, constructed for $\mu_1=3,\,  \mu_2=0.2$, under the initial conditions
\begin{equation}
\label{eq:ini}
\begin{split}
\ell_0=0.2,\quad \varphi_{10}\in \{0.25\pi,0.5\pi, 0.75\pi\},\qquad \varphi_{20}= v_0=0=\omega_{10}=\omega_{20}.
\end{split}
\end{equation}
As the studied system spans six-dimensional phase space, there are six Lyapunov exponents in the whole spectrum denoted as $\Lambda=\{\lambda_1,\ldots, \lambda_6\}$ where $\lambda_1$ is the largest Lyapunov exponent. However, in the context of the Hamiltonian systems and their symplectic structure, these exponents are not independent, i.e., they obey the pairing rule. Hence if $\lambda$ is a Lyapunov exponent, then $-\lambda$ is also, which indicates   $\sum \lambda_i=0$. This is related to the time-reversibility of the system and Liouville’s
theorem [99]  concerning the volume-preserving of the phase space in conservative Hamiltonian systems. As we can notice, the integration time of 5000 units was sufficient to ensure the convergence of the
Lyapunov exponents. Moreover, due to the presence of the first integral, which is energy conservation, two Lyapunov exponents converge to zero. We consequently use this crucial property to estimate the sufficient convergence of the Lyapunov exponents spectrum. The above fact was missed in the paper~\cite{Li:18::}, where the authors did not sufficiently ensure the convergence of Lyapunov exponents during the study of a Hamiltonian system, which implied erroneous results. We later corrected this result in~\cite{Szuminski:20b::}. 
Thus, the possible Lyapunov exponents spectrum
of the
system~\eqref{eq:rhs}, can be of the form  $\Lambda=\{\lambda_1,\lambda_2,\lambda_3,-\lambda_3,-\lambda_2,-\lambda_1\}$, where $\lambda_3\approx 0$.

In Fig.~\ref{fig:parki0}, we observe that for a given initial condition~\eqref{eq:ini} with relatively small initial amplitudes, such as $\varphi_{10}=0.25\pi$ or $\varphi_{10}=0.5\pi$, the system dynamics is very regular and non-chaotic, as indicated by all lambdas tending to zero. However, for Hamiltonian systems, computations of Lyapunov exponents for a very small number of initial conditions may not provide a comprehensive view of the system’s dynamics and integrability. For example, in Fig.~\ref{fig:parki0}(c), we can see significantly different structures of the Lyapunov exponents in the visible spectrum compared to those computed for smaller initial amplitudes. It is evident that for a sufficiently large initial swing angle, $\varphi_{10}=0.75\pi$, the system's dynamics is hyperchaotic with two positive Lyapunov exponents.

\begin{figure*}[htp]
		\centering
					\subfigure[The initial condition:	$\ell_0=0.2,\, \varphi_{10}=0.25\pi, \quad \varphi_{20}= v_0=0=\omega_{10}=\omega_{20}$]{		
					\includegraphics[width=0.38\linewidth]{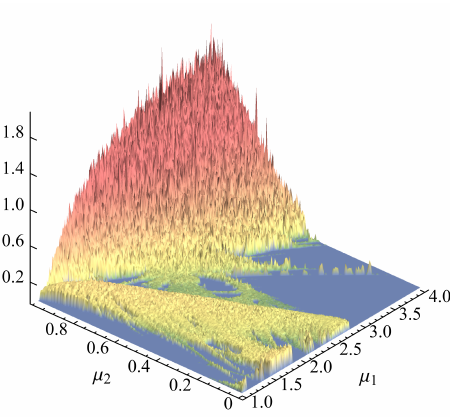}\hspace{-0.2cm}
		\includegraphics[width=0.31\linewidth]{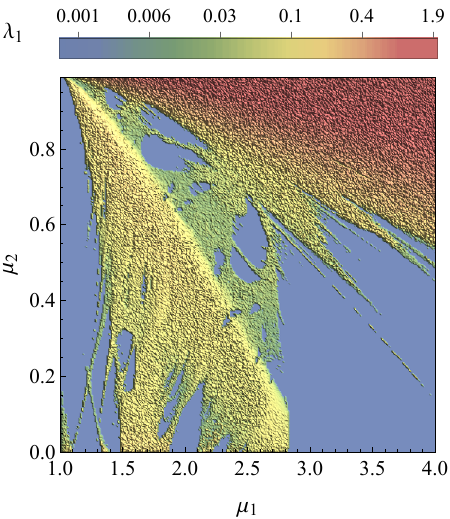}	 	\includegraphics[width=0.31\linewidth]{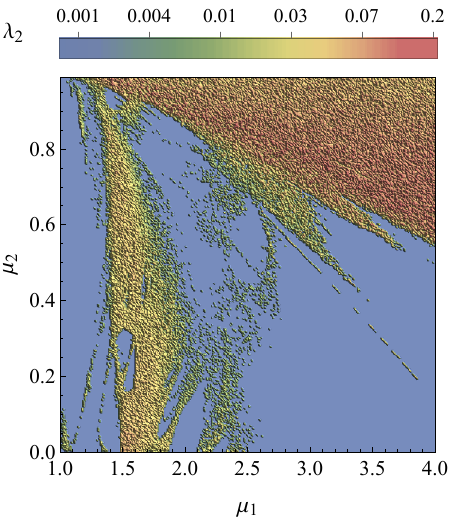}}		 
		\subfigure[The initial condition:		$\ell_0=0.2,\, \varphi_{10}=0.5\pi, \quad \varphi_{20}= v_0=0=\omega_{10}=\omega_{20}$]{		\includegraphics[width=0.38\linewidth]{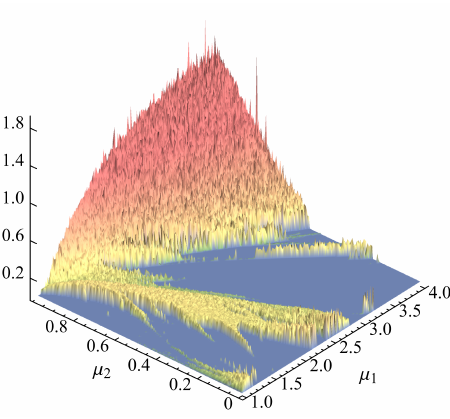}\hspace{-0.2cm}
		\includegraphics[width=0.31\linewidth]{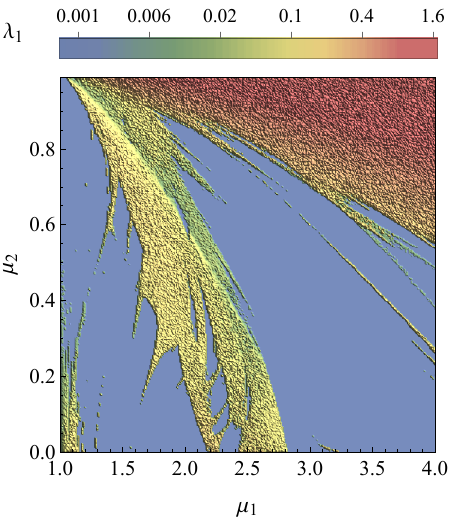}	 	\includegraphics[width=0.31\linewidth]{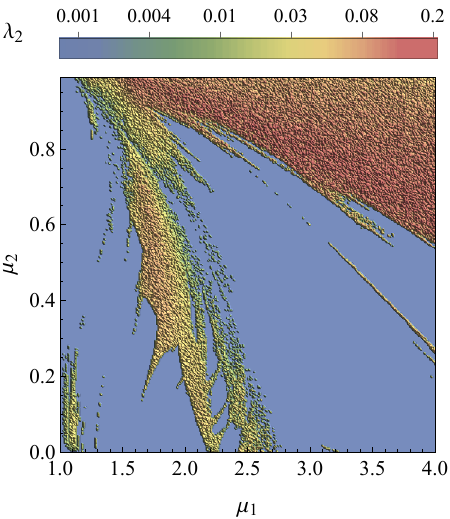}		 	}
		\subfigure[The initial condition:	$\ell_0=0.2,\, \varphi_{10}=0.75\pi, \quad \varphi_{20}= v_0=0=\omega_{10}=\omega_{20}$]{
		\includegraphics[width=0.38\linewidth]{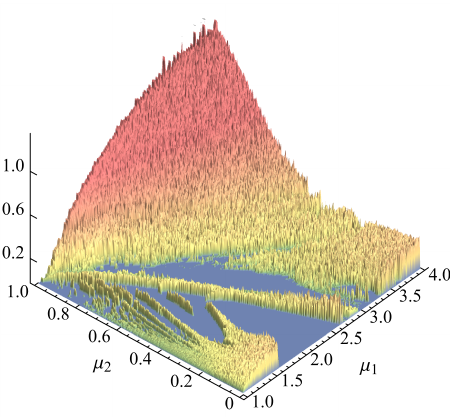}\hspace{-0.2cm}
		\includegraphics[width=0.31\linewidth]{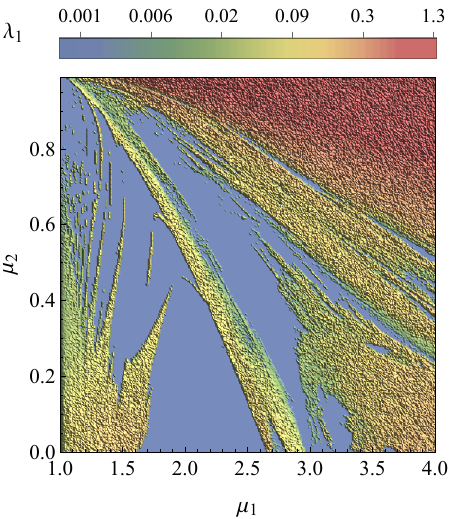}	 	\includegraphics[width=0.31\linewidth]{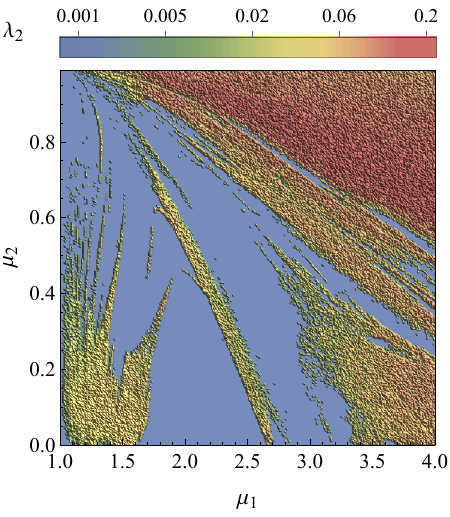}		 			}			\caption{ (Color online) Three-dimensional Lyapunov's exponents' diagrams of the system~\eqref{eq:rhs} depicted in $(\mu_1,\mu_2,\lambda)$-space, and the projections of $\lambda_1$ and $\lambda_2$ onto $(\mu_1,\mu_2)$-plane. The colorful diagrams were obtained by numerical integrations of Lyapunov's exponent's spectra on a grid of $400\times 400$ values of the parameters $(\mu_1,\mu_2)$, under the respective initial conditions {eq:ini}. The color scales are logarithmic, corresponding to the magnitudes of $\lambda_1$ and $\lambda_2$, respectively.  \label{fig:parki}}
				
					\end{figure*}							
We can construct the Lyapunov exponents diagram in a parameters plane by iterating the above procedure for different parameter values or initial conditions. Fig.~\ref{fig:parki} show a three-dimensional diagrams of Lyapunov exponents~$(\lambda_1,\lambda_2,\lambda_3)$ as a function of the parameters $(\mu_1, \mu_2)$.  On the rights, the projections of $\lambda_1$ and $\lambda_2$ onto the plane $(\mu_1, \mu_2)\in \R^{2}_+$ is given. Due to visualization purposes, the color scale is logarithmic, defined by the formula: $\text{scalling}(x)=\frac{\log(x/\text{min})}{\log(\text{max})/\log(\text{min})}$, where $\text{min}$ and $\text{max}$ represent the minimum and maximum values of the Lyapunov exponent, respectively. The above enables the detection of the regions with weak chaotic properties where the   Lyapunov exponents are relatively small, which could be missed using the standard linear scaling. In the studied system, chaos typically occurs when the Lyapunov exponent exceeds $0.003$. These colorful diagrams, as depicted in  Fig.~\ref{fig:parki}, were obtained by numerical computations of Lyapunov exponents $\lambda_1$ and $\lambda_2$ on a uniform grid of $400\times 400$ values of the parameters $\mu_1\in(1,4)$ and $\mu_2\in (0,1)$ with the initial conditions~\eqref{eq:ini}. Thus, the presented diagrams show how the change in parameter values affects the system's dynamics. The blue areas correspond to regular, non-chaotic oscillations of the pendulums. At the same time, the rest of the domain is responsible for the chaotic motion with two positive Lyapunov exponents. Indeed, the Lyapunov diagrams of exponents $\lambda_1$ and $\lambda_2$ mostly coincide confirming the
hyperchaotic nature of the studied model.

Let us focus mainly on the first diagram in Fig.~\ref{fig:parki}(a), in which chaos occupies about $60\%$ of the available area. From this diagram, 
we can easily estimate how the strength of chaos and, in fact, hyperchaos varies with the changes of the parameters $(\mu_1,\mu_2)$.  For $\mu_1=1$, the system dynamics are regular, and all trajectories are terminating - the pendulums are sliding down with low amplitudes of oscillations. It is the well-established property of the swinging Atwood machine, where the counterweight mass is less or equal to the oscillating masses~\cite{Tufillaro:85::,Tufillaro:88a,Tufillaro:95::}.  
 Moreover, we can observe a narrow rectangle $(1,1.05)\times (0,1)$, which also corresponds to values of the parameters for which the system's motion is regular with quasi-periodic or periodic orbits. The detailed distinction between these two types of motion will be provided in the next subsection. We know that for $\mu_2=0$, the system reduced to the classical SAM model of two degrees of freedom. Thus, we may expect the regular behavior of the system under the perturbation $\mu_2\approx 0$. Our suspicion is generally not confirmed. Instead of that, we note that for $\mu_2\approx 0$, there is a vast range of values of $\mu_1$ for which the system's motion is complicated and hyperchaotic.
 Moreover, we can observe many regular gaps (seas) between chaotic layers. What is important is to note that for $\mu_1=3$, the double-pendulum Atwood machine is generally not integrable. Both lambdas are positive for a wide range of $\mu_2$, for instance, $\mu_2\in(0.5,1)$, which precludes the integrability and the presence of the additional first integral inside the system~\eqref{eq:rhs}. The last but not least important observation from the studied Lyapunov exponent diagram is the fact that in the chaotic domains of the diagram, the strength of chaos is proportionally increasing with the increasing values of $\mu_2$ as well. This is because the presence of the second pendulum greatly impacts the system dynamics. Indeed, the largest Lyapunov exponent reaches its maximum value $\lambda_1\approx 2.1$  at $\mu_2 \to 1$.

Figs.~\ref{fig:parki}(b)-\ref{fig:parki}(c) show the additional Lyapunov diagrams with increasing values of the initial swing angle $\varphi_{10}=0.5\pi$ and $\varphi_{10}=0.75\pi$. Although these diagrams differ from the one given in Fig.~\ref{fig:parki}(a), they possess some similarities in the general structures. However, it is surprising that for a higher value of the initial amplitude 
$\varphi_{10}=0.5\pi$, the chaos occupies smaller percentege ($41 \%$) of the area of the parameter plane $(\mu_1,\mu_2)$.   Moreover, the small regular rectangle $(1,1.05)\times (0,1)$ visible in  Fig.~\ref{fig:parki}(a) now mostly decays, and we observe the appearance of chaos within this region. This fact is even better visible when we increase the value of $\varphi_{10}$, e.g.,  up to 
$\varphi_{10}=0.75\pi$. Nevertheless, despite the large value of the initial swing amplitude, we can detect several values of $(\mu_1,\mu_2)$ for which system motion is non-chaotic. For instance, in the studied Lyapunov diagrams, we can notice that for $\mu_1=3$ and $\mu_2\approx 0$, the system dynamics is very regular. In this scenario, the studied system can be treated as the integrable swinging Atwood machine model with a small perturbation coming from the oscillations of the second pendulum with a small value of the mass $m_2$. A detailed study of this situation will be given in the next subsection. 
\begin{figure}[t]
		\centering
		\subfigure[The mass ratios: $\mu_1=3,\, \mu_2=2\cdot 10^{-4}$]{
		\includegraphics[width=0.44\linewidth]{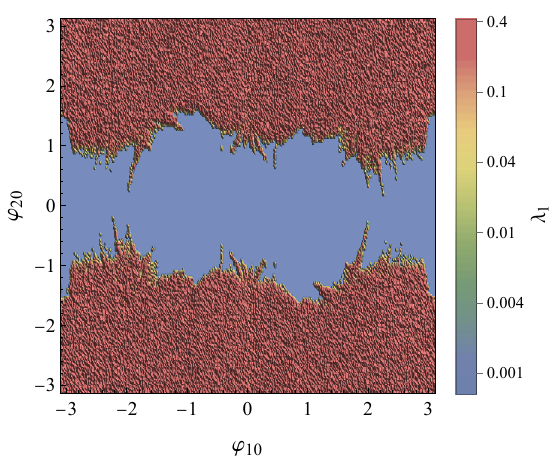}} \hspace{0.4cm}
		\subfigure[The mass ratios: $\mu_1=3,\,\mu_2=0.2$]{
		\includegraphics[width=0.44\linewidth]{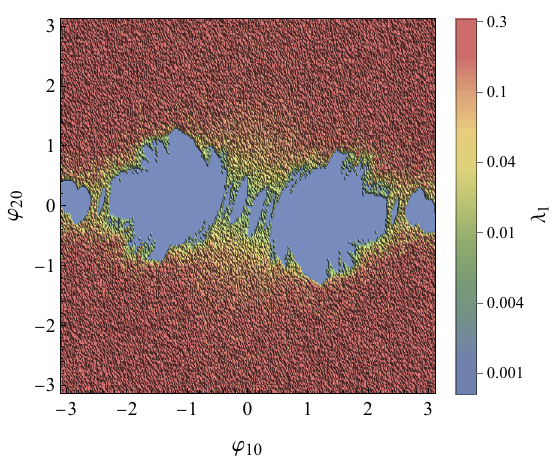}}	\\ \vspace{-5pt}
		\subfigure[The mass ratios: $\mu_1=3,\,\mu_2=0.5$]{
		\includegraphics[width=0.44\linewidth]{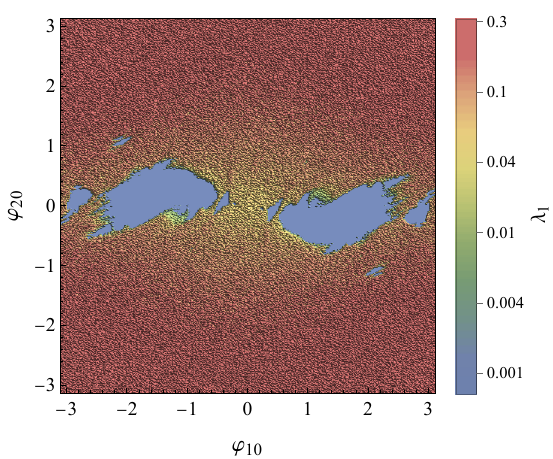}} \hspace{0.4cm}
		\subfigure[The mass ratios: $\mu_1=3,\,\mu_2=0.9$]{	\includegraphics[width=0.43\linewidth]{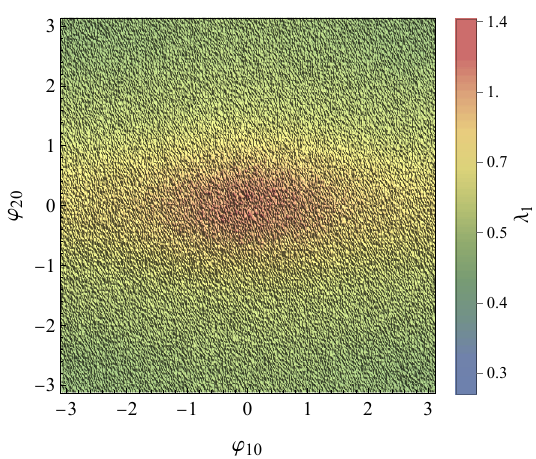} }
					\caption{The largest Lyapunov exponent in the plane of the initial swing angles $(\varphi_{10},\varphi_{20})$, constructed for $\mu_1=3$ with varying $\mu_2$. The numerical
integrations were performed successively for a uniform grid of $400\times 400$ values of $(\varphi_{10},\varphi_{20})\in(-\pi,\pi)$, with $\ell_0=0.2$ and zero initial velocities. The color scale is logarithmic, corresponding to the magnitude of $\lambda_1$. The plots visualize two zones: regular and chaotic. Blue regions with $\lambda_1\approx 0$ indicate regular dynamics, while the rest of the domain is responsible for the system’s chaotic behavior.  }
		\label{fig:katy}
	\end{figure}

The computation of Lyapunov exponents diagrams can provide an important insight into the system's dynamics by plotting the largest Lyapunov exponent as a function of the initial conditions of state variables. Therefore, it allows for both a qualitative and quantitative description of chaos.
In Fig.~\ref{fig:katy}, we show the Lyapunov exponents diagrams on the plane of initial swing angles $(\varphi_{10},\varphi_{20})$, constructed for 
\begin{equation}
\label{eq:parki2}
\mu_1=3,\qquad \mu_2\in\{2\cdot 10^{-4},0.2,0.5, 0.9\},
\end{equation}
with the initial conditions
 \begin{equation}
 \label{eq:init}
 \ell_0=\frac{1}{5},\quad \varphi_{10}\in(-\pi,\pi),\quad \varphi_{20}\in(-\pi,\pi),\quad \quad v_0=\omega_{10}=\omega_{20}=0,
 \end{equation}
where $\varphi_{10}$ and $\varphi_{20}$ are the control parameters. These colorful diagrams were obtained by the successive computations of the largest Lyapunov exponent $\lambda_1$ for a uniform grid of $400\times 400$ values of the control parameters $(\varphi_{10},\varphi_{20})$.  The diagrams depicted in Fig.~\ref{fig:katy} show how the change in the initial swing angles of the pendulums with zero initial velocities affects the system's dynamics with varying parameter values $\mu_2$. As expected, the diagrams possess reflection symmetry.   The plots visualize two behaviors: regular and chaotic. Blue regions with $\lambda_1\approx 0$ indicate regular dynamics, while the rest of the domains are responsible for the system’s chaotic motion. In Fig.\ref{fig:katy}(a), we present the Lyapunov diagram computed for an infinitesimally small mass ratio of $\mu_2 = 2 \cdot 10^{-4}$. This configuration arises from a slight perturbation of the integrable SAM model, achieved by attaching a second pendulum with mass $m_2$, which is five orders of magnitude smaller than the mass $m_1$, to its end. As expected, for sufficiently small initial amplitudes of the second pendulum $\varphi_{20}$, the system exhibits regular (non-chaotic) oscillations for all initial angles $\varphi_{10}$ of the first pendulum. However, for larger values of the initial swing angle, specifically $\varphi_{20} > 0.25\pi$, chaotic motion emerges for some values of $\varphi_{10}$, and $\varphi_{20} > 0.5\pi$, chaos manifests for every $\varphi_{10}$. When the mass ratio is increased to $\mu_2 = 0.2$, the region corresponding to chaotic motion expands, as shown in Fig.\ref{fig:katy}(b).
Now, the chaos occupies almost $80\%$ of the admissible region. Surprisingly, in the central part of the diagram, where 
values of $(\varphi_{10},\varphi_{20})$ correspond to small initial amplitudes of the pendulums, the regular area divergence, and we observe the appearance of weak chaotic behavior of the system.
Moreover, we can notice that during the computations of the Lyapunov exponents spectrum, as evidenced in Fig.~\ref{fig:parki0}, we did not catch the proper initial condition for the system's motion to be chaotic. It is clear that the values $(\varphi_{10},\varphi_{20})\in \{(0,0.25\pi),(0,0.5\pi)\}$ fall within the regular region of the Lyapunov exponents diagram, as depicted in Fig.~\ref{fig:katy}(b). Furthermore, increasing the mass ratio to $\mu_2=0.5$, indicating that the pendulum masses are equal, results in a more intricate Lyapunov exponent diagram with small regular gaps between chaotic layers. Ultimately, for $\mu_2=0.9$, where the mass of the second pendulum is nine times greater than the first one, the system dynamics become entirely ergodic, as shown in Fig.~\ref{fig:katy}(d),  with the largest Lyapunov exponent reaching a maximum value of $\lambda_1=1.45$.

		\begin{figure}[t]
		\centering
		\subfigure[The mass ratios: $\mu_1=3,\,\mu_2=2\cdot 10^{-4}$]{
		\includegraphics[width=0.47\linewidth]{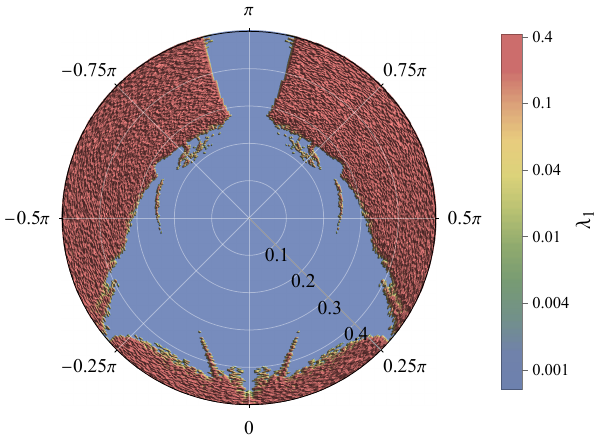}  } \hspace{5pt}
		\subfigure[The mass ratios: $\mu_1=3,\,\mu_2=0.2$]{
		\includegraphics[width=0.47\linewidth]{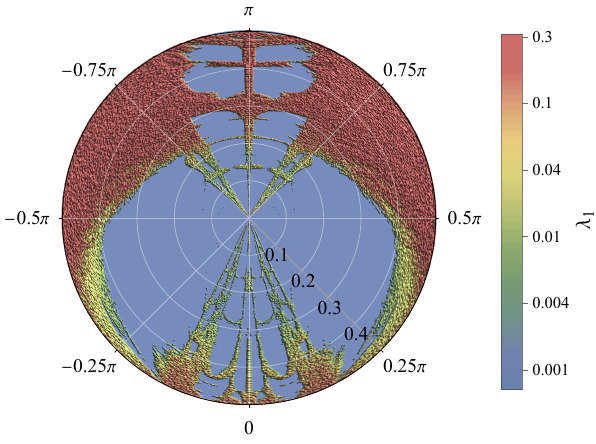}	} \\
		\subfigure[The mass ratios: $\mu_1=3,\,\mu_2=0.5$]{\includegraphics[width=0.47\linewidth]{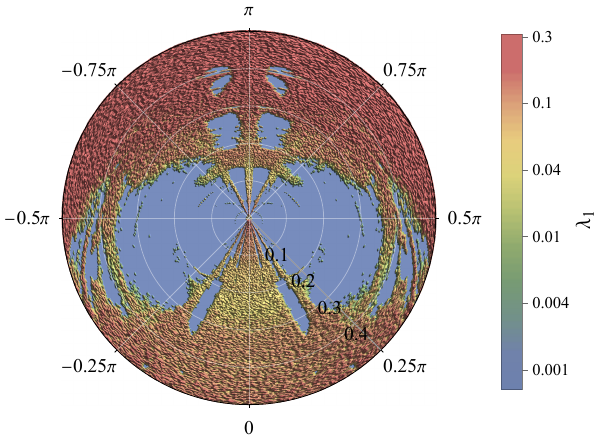}  }\hspace{5pt}
\subfigure[The mass ratios: $\mu_1=3,\,\mu_2=0.9$]{
	\includegraphics[width=0.47\linewidth]{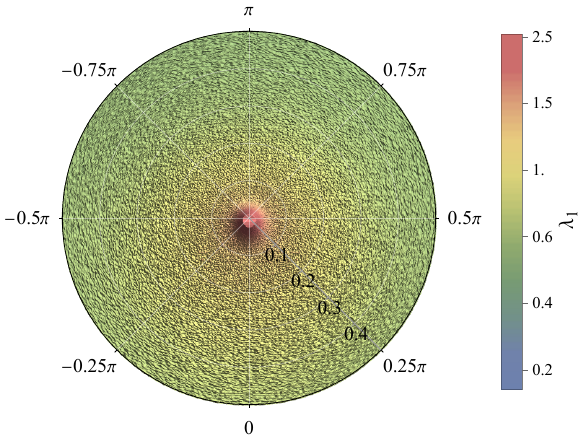}	 		}					 \caption{(Color online)  The Lyapunov diagrams for system~\eqref{eq:rhs} in the polar plane $(\ell_0, \varphi_{10})$ with $\varphi_2=0$ are constructed for the initial conditions~\eqref{eq:ini_2}, with varying values of $\mu_2$. In the radial direction, we measure $\ell_0 \in (0, 0.5]$, and in the angular direction, we measure $\varphi_{10} \in (-\pi, \pi)$. The color scale is logarithmic, corresponding to the magnitude of $\lambda_1$. The plots reveal two zones: regular and chaotic. Blue regions indicate regular dynamics, while regions with $\lambda_1 > 0$ correspond to the system’s hyperchaotic behavior. \label{fig:lyap_polar2}}
	\end{figure}
	
	\begin{figure}\centering
	\subfigure[The initial length: $\ell_0=0.2$]{
	\includegraphics[width=0.45\linewidth]{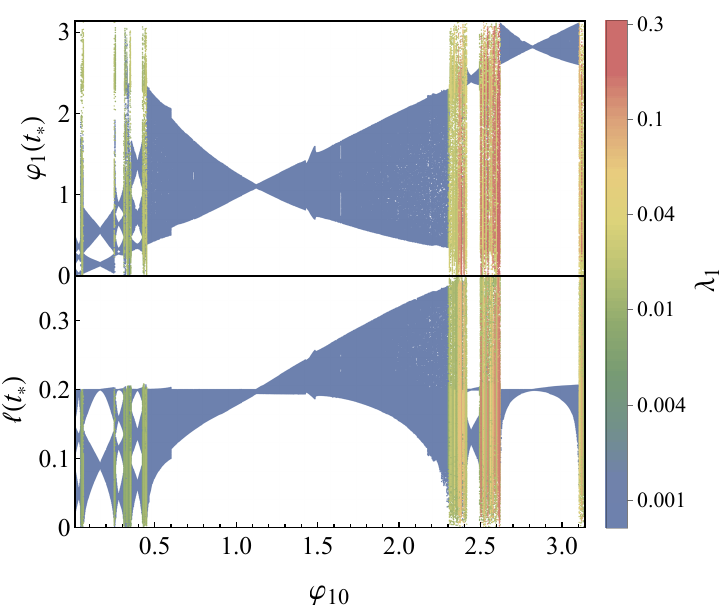}}
	\subfigure[The initial swing angle: $\varphi_{10}=\frac{\pi}{8}$]{
	\includegraphics[width=0.45\linewidth]{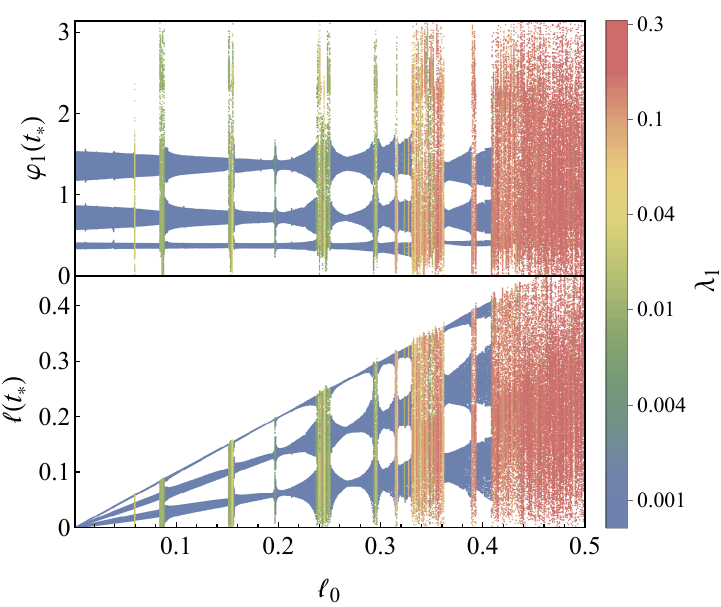} }\caption{(Color online)  Phase-parametric diagrams of system~\eqref{eq:rhs} versus the initial swing angle $\varphi_{10}\in(0,\pi)$ and the ratio of the initial lengths of the pendulums $\ell_0\in (0,0.5]$. Initial conditions and values of the parameters are taken from Fig.~\ref{fig:lyap_polar2}(b). There are two cases to consider: a) we move in the angular direction of the Lyapunov diagram~\ref{fig:lyap_polar2}(b)  and $\varphi_{20}\in(0,\pi)$; b) we choose the initial swing angle to be $\varphi_{20}=\pi/8$ and we move in the radial direction of the Lyapunov diagram for $\ell_0\in(0,0.5]$. Here $\varphi_1(t_\star)$ and $\ell(t_\star)$ denotes values of the state variables$\varphi_1,\ell$ when the first pendulum reaches its maximal values, i.e. when $\omega_1(t_\star)=0$ and $\omega_1'(t_\star)<0$ for some $t_\star$.
The diagram is combined with the largest Lyapunov exponent, $\lambda_1$, and the color scale is logarithmic, corresponding to the magnitude of $\lambda_1$. A very good agreement between the phase-parametric diagram and the Lyapunov diagram (see Fig.~\ref{fig:lyap_polar2}(a)) is observed. We observe the coexistence of periodic, quasi-periodic, and chaotic orbits, along with periodic windows between chaotic layers. Exemplary periodic, quasi-periodic, and chaotic orbits are plotted in Fig.~\ref{fig:periodicki0}.}
	\label{fig:bifek2}
\end{figure}
\begin{figure*}[t]	\centering
	\subfigure[$\varphi_{10}=1.15$, periodic trajectory]{
	\includegraphics[width=0.32\linewidth]{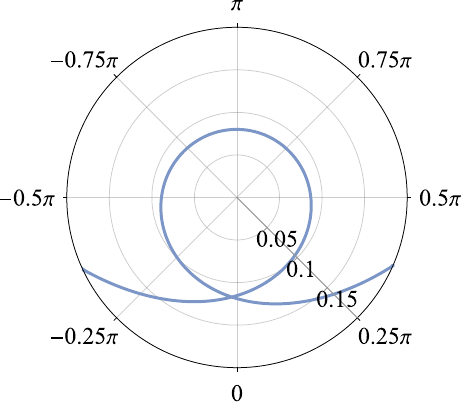}}
		\subfigure[$\varphi_{10}=0.5\pi$, quasi-periodic trajectory]{
	\includegraphics[width=0.32\linewidth]{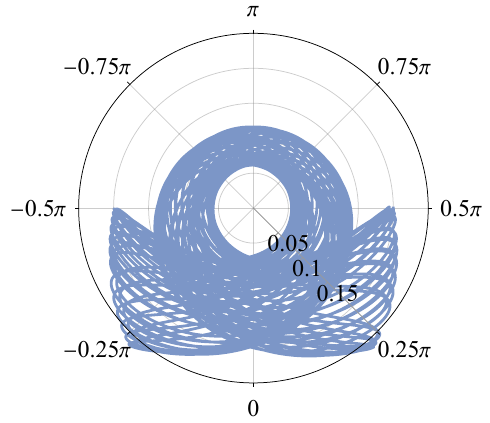}}
	\subfigure[$\varphi_{10}=0.75\pi$, chaotic trajectory]
{
	\includegraphics[width=0.32\linewidth]{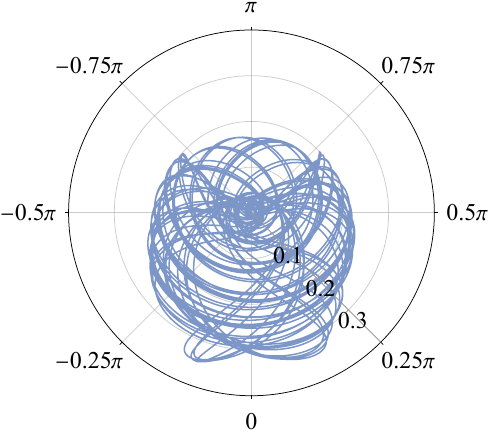}}
	\caption{(Color online)  Polar plots of exemplary periodic, quasi-periodic, and chaotic trajectories of the variable-length pendulum. Values of the parameters and the initial conditions were taken from the phase-parametric diagram presented in Fig.~\ref{fig:bifek2}(a). In the radial direction, we measure $\ell(t)$; in the angular direction, we measure $\varphi_1(t)$.}
	\label{fig:periodicki0}
\end{figure*}

\begin{figure}[t]
	\centering
	\includegraphics[width=0.99\linewidth]{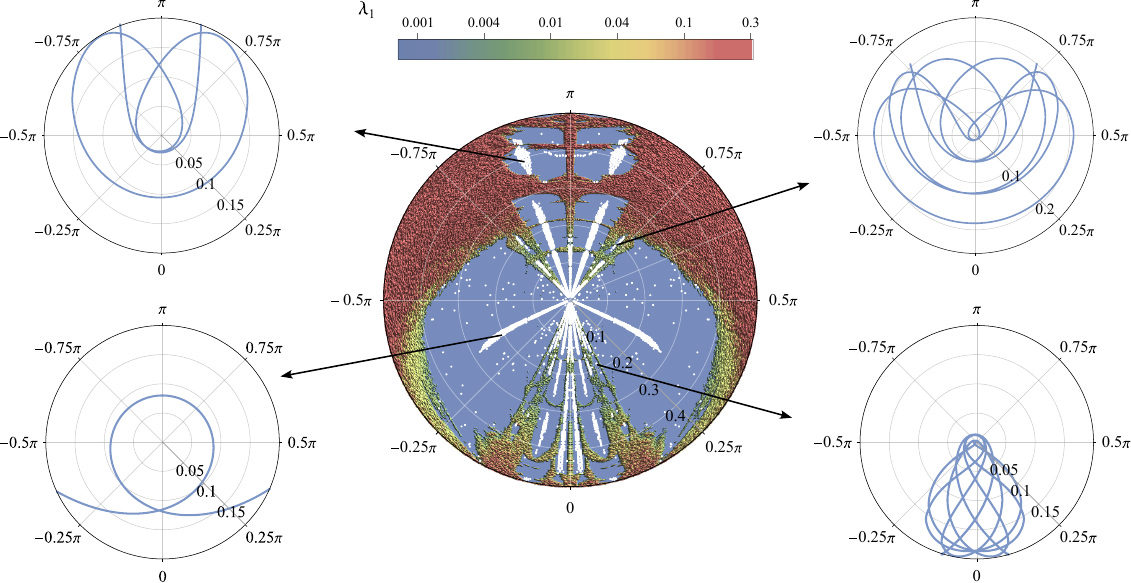}
	\caption{(Color online)  The  Lyapunov refined map formed by the Lyapunov diagram
	 presented in Fig.~\ref{fig:lyap_polar2}(b) with depicted points corresponding to values of $(\ell_0,\varphi_{10})$ for which motion of the variable-length pendulum is periodic. 	 	The polar plots computed for distinct initial conditions, taken from different periodic clusters of the map, are given.    In the radial direction, we measure $\ell(t)$; in the angular direction, we measure $\varphi_1(t)$. The different structures of the trajectories, as well as their periodicity, are visible. \label{fig:bif_pik}}
\end{figure}

\begin{figure*}[t]	\centering
\subfigure[$\ell_0=0.33, \, \varphi_{10}=0.41, \, \varphi_{20}=0$]{
		\includegraphics[width=0.32\linewidth]{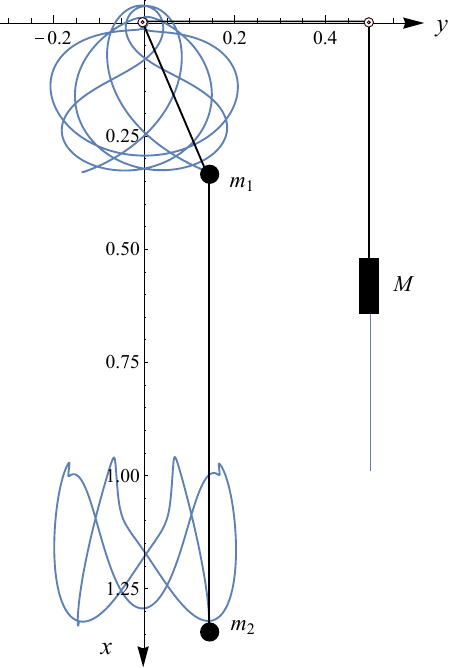}	}\subfigure[$\ell_0=0.22, \, \varphi_{10}=1.13, \, \varphi_{20}=0$]{\includegraphics[width=0.33\linewidth]{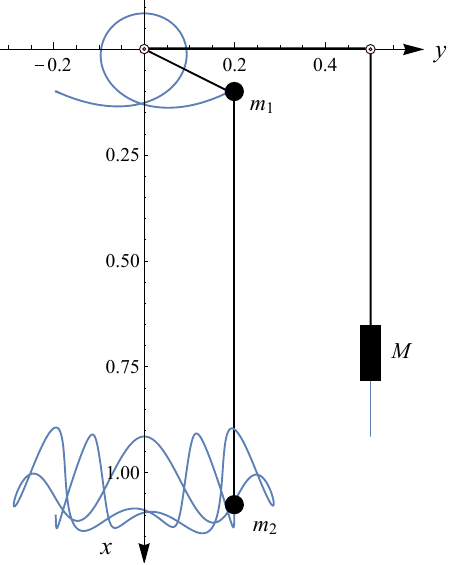}}
		\subfigure[$\ell_0=0.33, \, \varphi_{10}=2.84, \, \varphi_{20}=0$]{
			\includegraphics[width=0.32\linewidth]{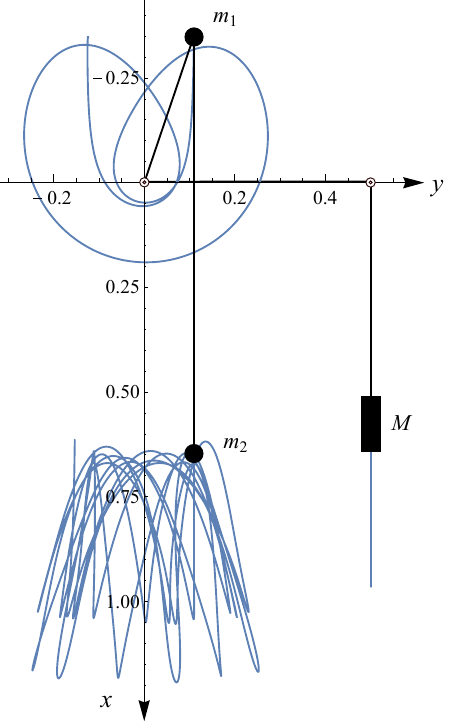}}		\caption{
(Color online)  The periodic behavior of the whole system with depicted trajectories plotted in the Cartesian plane. 
	\label{fig:periodic}}
			\end{figure*}

In Fig.~\ref{fig:lyap_polar2}, we present the polar plots of Lyapunov exponents diagrams of system~\eqref{eq:rhs} computed for the same  values of the parameters~\eqref{eq:parki2} as in Fig.~\ref{fig:katy}, but choosing the different set of initial conditions
\begin{equation}
\label{eq:ini_2}
\begin{split}
\ell_0\in(0, 0.5),\quad \varphi_{10}\in (0,\pi),\qquad \varphi_{20}= v_0=0=\omega_{10}=\omega_{20}.
\end{split}
\end{equation}
Thus,    $\ell_0$ and $\varphi_{10}$ are now the control initial state variables, indicating how the change of the initial length of the variable-length pendulum and its initial swing angle affects the system's dynamics.
 The colorful diagrams shown in Fig.~\ref{fig:lyap_polar2} were created by numerically computing Lyapunov’s exponents on a grid of $400\times 400$ values of $(\ell_0,\varphi_{10})$. These values were then plotted in the $(x,y)$-Cartesian plane. The radial direction measures the rescaled initial length of the pendulum, $\ell_0$, while the angular direction represents the values of its initial amplitude $\varphi_{10}$. The Lyapunov diagrams visualize several zones: regular and complex ones. Blue regions indicate regular dynamics with periodic and quasi-periodic solutions, while regions with $\lambda_1 > 0$ correspond to the system’s chaotic behavior. It is observed that for small $\mu_2=2\cdot 10^{-4}$ and $\ell_0<0.2$, the system exhibits regular (non-chaotic) oscillations for every initial swing angle $\varphi_{10}\in(-\pi,\pi)$. However, as $\ell_0$ is increased, the first signs of chaotic motion occur around the point $(\ell_0,\varphi_{10})=(0.25,0.5\pi)$. It is then noticed that as $\ell_0$ increases, the necessary value of the initial swing angle $\varphi_{10}$ for chaos to appear decreases, as expected. Interestingly, there is a large, regular gap between chaotic layers with $|\varphi_{10}|> 0.9\pi$, where chaos is not observable even for large initial $\ell_0$. This behavior is somewhat unexpected and unique in the context of Hamiltonian pendulum systems, as such large energy values typically lead to chaotic dynamics throughout the system.
 
  The remaining polar plots of Lyapunov exponent diagrams, as depicted in Fig.~\ref{fig:lyap_polar2}(b)-(d), show the relation between initial values $(\ell_0,\varphi_{10})$ with further increasing $\mu_2$, so the impact to the dynamics of the system coming from the second pendulum becomes noticeable. As we could expect, for $\mu_2=0.2$, the situation is far different compared to Fig.~\ref{fig:lyap_polar2}(a). Now, we can observe the appearance of chaos for initial swing amplitudes close to zero. Indeed, there are some values of $\varphi_{10}$, where, moving in the radial direction, all visible values of $\ell_0$ are responsible for the chaotic motion, e.g. $\varphi_{10}=0.1\pi$. Naturally, when $\ell_0$ increases, the strength of chaos also increases. For $\ell_0\approx 0.5$, the system's motion is chaotic for every $\varphi_1\in(-\pi,\pi)$. Moreover, the regular gap with $|\varphi_{10}|>0.9 \pi$ visible in the previous diagram now splits into the smaller, connected domains. When moving in the radial direction in this case, we observe five regular seas intersected by narrow chaotic layers. For a higher value of the mass ratio $\mu_2=0.5$, which implies that the masses $m_1$ and $m_2$ of the swinging pendulums are equal,
   the chaos occupies a larger area of the  Lyapunov diagram. At the same time, for $\mu_2=0.9$, there are no signs of regular oscillations of the pendulums, as the entire permissible region in the diagram exhibits strong chaotic behavior, where the largest Lyapunov exponent reaches its maximal value, up to $\lambda_1=2.51$.

\subsection{Phase-parametric diagrams}
The method of Lyapunov exponents is a fundamental tool in the study of dynamical systems, providing a quantitative description of chaos and its strength. Indeed, this method enables us to detect the coexistence between the chaotic and regular motions, as evidenced in the previously discussed Figs.~\ref{fig:parki}-\ref{fig:lyap_polar2}. However,  in the framework of Hamiltonian systems, it has limitations, particularly when it comes to distinguishing periodic motion from quasi-periodic motion. More technically, the method of Lyapunov exponents does not capture the topological differences between periodic and quasi-periodic solutions.
For instance, in the blue regions of the studied Lyapunov diagrams, when $\lambda_1\approx 0$, we cannot distinguish these two types of motion.  	The above leads to incomplete information about the system’s phase space structure in these cases.  This weakness highlights the need for supplementary techniques to fully understand the nature of regular motion in Hamiltonian dynamics. As has already been shown in our previous works~\cite{Szuminski:23::,Szuminski:24::}, the computation of the so-called phase-parametric diagrams is effective for these purposes.
The phase-parametric diagram effectively illustrates how a system's qualitative behavior evolves as a certain control parameter is varied. 
For a given Hamiltonian system, bifurcation points mark critical values of parameters where the system undergoes a qualitative shift in its dynamics. These changes can include the transition from regular to chaotic motion, frequently called the route to chaos; detection of periodic orbits, their numbers, and stability; regular gaps (windows) between chaotic regions in chaos or symmetry breaking~\cite{Stachowiak:06::,Stachowiak:15::}.
	
Fig.~\ref{fig:bifek2} illustrates the phase-parametric diagrams of the system computed for two one-parameter families of initial conditions taken from Fig.~\ref{fig:lyap_polar2}(b). In Fig.~\ref{fig:bifek2}(a), we move in the angular direction of the Lyapunov diagram with $\ell_0=0.2$ and $\varphi_{10}\in(0,\pi)$, while  Fig.~\ref{fig:bifek2}(b) corresponds to the case that we fix $\varphi_{10}=0.125\pi$, and we move in the radial direction of the Lyapunov diagram over the range $\ell_0\in(0,0.5]$. In the studied phase-parametric diagrams, we display the dependence of $\varphi_1(t)$ and $\ell(t)$, when the first pendulum reaches its maximal values (amplitudes), with the change of the control parameters $\varphi_{10}$ and $\ell_0$.  That is, for a given initial condition, we numerically integrate the equations of motion~\eqref{eq:rhs} and
build the diagrams by collecting points $\varphi_1(t_\star)$ and $\ell(t_\star)$
repeatedly, when $\omega_1(t_\star)=\dot \varphi_1(t_\star)=0$  and $\omega_1'(t_\star)<0$, for a certain $t_\star$. 
Moreover, to improve the analysis, we combine the calculated phase-parametric diagrams with the largest Lyapunov exponent previously computed in Fig.~\ref{fig:lyap_polar2}(b). First, we notice the excellent correspondence of the phase-parametric diagrams with $\lambda_1$ as evidenced in Fig.~\ref{fig:bifek2}. The blue regions of the phase-parametric diagram are responsible for regular oscillations, while the rest of the area is chaotic. 
Although Lyapunov exponent diagrams do not reveal it within regular regimes, we can now distinguish periodic orbits from quasi-periodic ones. For better understanding, we provide in Fig.~\ref{fig:periodicki0} three polar plots of periodic solutions, where the initial values of $\varphi_{10}$ were taken from the phase-parametric diagram illustrated in Fig.~\ref{fig:bifek2}(a). 
%It is worth mentioning that when we speak about periodic solutions, we refer to the first pendulum's regular oscillations. In a double pendulum system, if the first pendulum swings back and forth periodically, the second pendulum may show quasi-periodic behavior.

As we can notice, the phase-parametric diagram in~Fig.~\ref{fig:bifek2}(a) starts with small, very close to periodic oscillations of the first pendulum near the equilibrium. However, this solution seems to be unstable because in the right neighborhood of this point on the diagram, to be precise at $\varphi_{10}=0.05$, we observe the large, narrow pick of the chaotic behavior with $\lambda_1\approx 0.01$, which is quite surprising since the initial swing amplitude is very small. Furthermore, we observe the emergence of classical periodic windows between chaotic layers. Next, we identify a large regular gap over the range $\varphi_{10} \in [0.46, 2.3]$, with only one periodic solution, which further diverges into global chaos, where the maxima of $\varphi_1$ densely occupy the large admissible region of the diagram.

Fig.~\ref{fig:bifek2}(b) presents the second pair of the phase-parametric diagrams with $\ell_0\in(0,0.5]$ treated as the control parameter. Hence, these diagrams correspond to the situation in that we fix  $\varphi_{10}=0.125\pi$, and we move the radial direction of the Lyapunov exponents diagram depicted in Fig.~\ref{fig:lyap_polar2}(b). One can observe the complex coexistence of periodic (or close to periodic) and chaotic behavior of the system visible in terms of regular gaps between the chaotic layers. As $\ell_0$ increases, the complexity and strength of chaos also increase. There are no visible regular regions for $\ell_0\in[0.405,0.5]$. Besides, it's important to note that further increasing of $\ell_0$ (not shown) leads to strong chaotic dynamics, as both the phase-parametric diagram and the Lyapunov exponents diagram become completely chaotic.

\subsection{The Lyapunov refined map \label{sec:refined}}
As we have already seen, there is no single universal tool for the numerical analysis of Hamiltonian systems. The method of Lyapunov exponents is used as an indicator of chaotic dynamics. At the same time, it does not distinguish values of the parameters for which the system's motion is periodic or quasi-periodic. For this purpose, we used the phase-parametric diagrams, which give qualitative information about the existence of periodic, or close to period one, solutions and the frequency of oscillations. Therefore, let us combine these two methods more systematically to create one global picture of the system's dynamics. We show the general procedure for the parameters corresponding to the previously studied  Lyapunov
diagram depicted in Fig.~\ref{fig:lyap_polar2}(b). 
This diagram was computed for a grid of $400\times 400$ values of $(\ell_0,\varphi_{10})$. From this list, we delete these pairs  $(\ell_0,\varphi_{10})$  such that $\lambda_1>0$. Thus, we receive a smaller grid of (let's say) $n$ pairs of initial conditions for which the system's motion is not chaotic. Next, for each element of this grid, we numerically integrate equations of motion, and we collect the points $(\ell(t_\star),\varphi_1(t_\star))$, when $\omega_1(t_\star)=0$ and $\omega_1'(t_\star)<0$, for a wide range of $t$. As the result, we obtain $n$ lists of multiple elements containing points $(\ell(t_\star),\varphi_1(t_\star))$. Then, we make the similarity check. Namely, in
each list, we look, with proper accuracy, for a pattern of repeating values of $(\ell(t_\star),\varphi_1(t_\star))$. If the motion is periodic and very close to periodic one, then we obtain a certain cyclical pattern. The above procedure can be effectively implemented in Mathematica. In this way, the rough but effective distinction between periodic from ordinary quasi-periodic solutions is possible.  The results of the computations are given in the diagram, as illustrated in Fig.~\ref{fig:bif_pik}, which we called the \textit{Lyapunov refined map}. Here, the white dots correspond to distinct pairs of $(\ell_0,\varphi_{10})$ for which the motion of the variable-length pendulum is periodic or very close to periodic ones. The polar plots computed for distinct initial conditions are depicted for better understanding. As we can notice in the Lyapunov refined map, the white dots arrange clusters of myriad periodic solutions, forming resonance-like curves. From the corresponding polar plots, we can identify that each cluster of periodic solutions corresponds to different structures of solutions (trajectories) and their frequency rations.  Thus, the computed Lyapunov refined map helps us to understand the system dynamics and identify intervals of the control parameters $\varphi_{10}$ and $\ell_0$ for which the variable-length pendulum motion is chaotic, quasi-periodic, or periodic with certain frequency ratios. 
Therefore,  it gives us the most qualitative and comprehensive view of the system dynamics.  

Let us emphasize that the Lyapunov-refined map, shown in Fig.~\ref{fig:bif_pik}, illustrates the periodic solutions for a variable-length pendulum with an initial swing angle $\varphi_{10}$ and initial length $\ell_0$. However, for our purposes, we have not specified whether the periodic solutions of the first pendulum imply periodicity in the second pendulum. In general, it seems that this does not have to be the case. Nevertheless, from our detailed analysis (not shown), we deduce that certain white dots in the Lyapunov-refined map indeed correspond to the periodic behavior of the entire system. In Fig.~\ref{fig:periodic}, the periodic behavior of the complete system is depicted, with trajectories of specific frequency ratios plotted in the Cartesian plane. The positions of the point masses $m_1$, $m_2$, and $M$ are shown for initial conditions $(\ell_0, \varphi_{10})$ selected from Fig.~\ref{fig:bif_pik}, under which the motion is  periodic.

 		\begin{figure*}[t]
		\centering
		\subfigure[The mass rations: $ \mu_1=3,\, \mu_2=2\cdot 10^{-4}$. Regular domain with non-chaotic dynamics]{
						\includegraphics[width=0.4\linewidth]{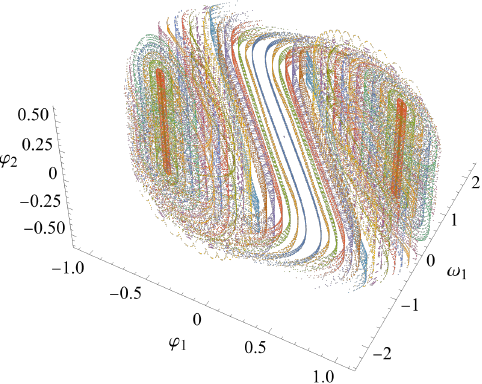}\hspace{1.5cm}\includegraphics[width=0.35\linewidth]{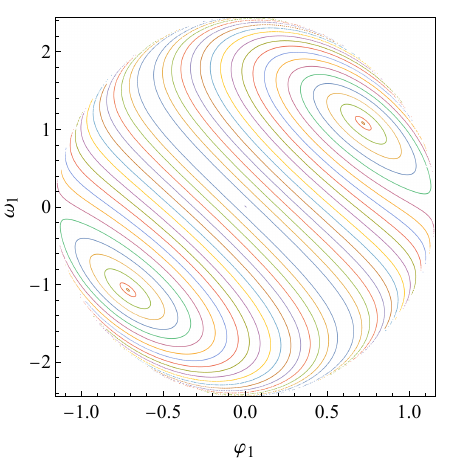}} \\
							\subfigure[The mass rations: $ \mu_1=3,\, \mu_2=0.2$. The complex domain with non-distinguishable dynamics]{
			\includegraphics[width=0.4\linewidth]{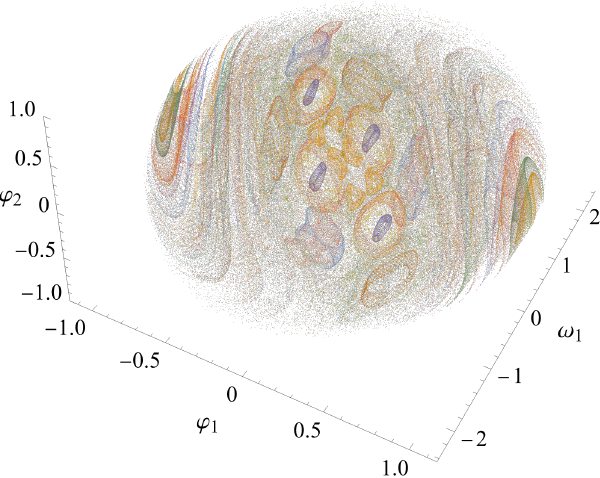}\hspace{1.5cm}\includegraphics[width=0.35\linewidth]{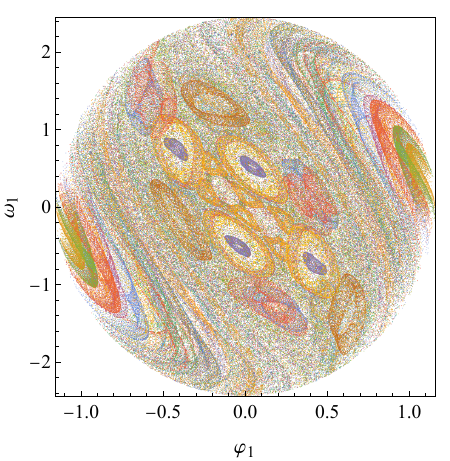}}
			\caption{The Poincar\'e sections of the system~\eqref{eq:rhs}, constructed for $\mu_1=3$ with varying $\mu_2$, at the constant energy level $E=E_\text{min}+0.12$. Here $E_\text{min}$ is the energy minimum at $\ell_0=0.2$ with zero values of the remaining  state variables. The cross-section was specified as $\ell=0.2$, with the direction $v>0$.  The results were projected onto  
			$(\varphi_1,\omega_1,\varphi_2)$-space, and to
			the plane $(\varphi_1,\omega_1)$. 
		\label{fig:ciecie_variable_0}}
	\end{figure*}
				\begin{figure*}[t]
		\centering
		\subfigure[The mass ratios $\mu_1=3,\, \mu_2=2\cdot 10^{-4}$. Regular, non-chaotic domain with two distinct  periodic solutions.]{
	\includegraphics[width=0.8\linewidth]{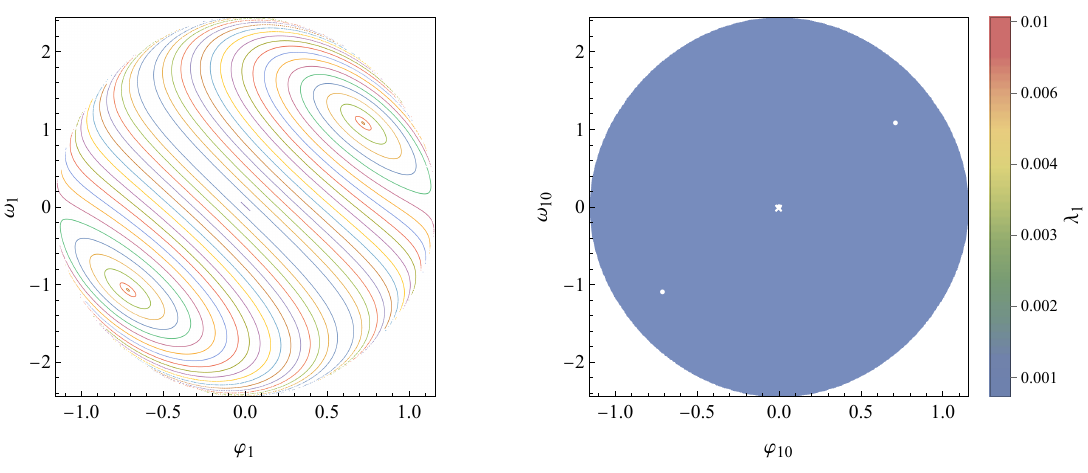}}
	
	\subfigure[The mass ratios $\mu_1=3,\, \mu_2=0.2$. Coexistence of  regular and chaotic areas with depicted periodic solutions.]{	\includegraphics[width=0.8\linewidth]{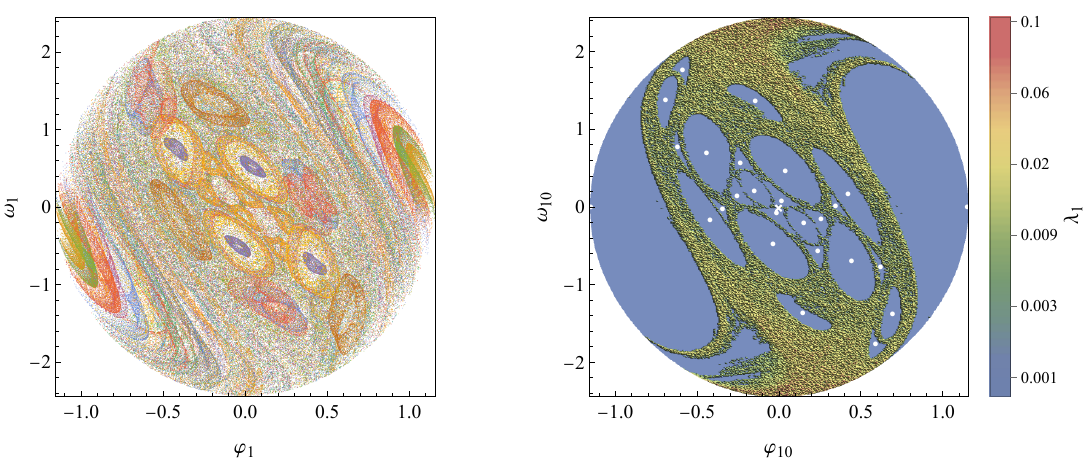}}
		\caption{(Color online) The Poincar\'e section of the system~\eqref{eq:rhs} and their corresponding  Lyapunov--Poincar\'e refined maps, constructed for constant values of the parameter $\mu_1=3$ with varying $\mu_2$, at the constant energy level $E=E_\text{min}+0.12$.  Here, $E_\text{min}$ is the energy minimum at $\ell_0=0.2$ with zero values of the remaining  state variables. The cross-section was specified as $\ell=0.2$, with the direction $v>0$. In the Lyapunov--Poincar\'e refined map, the color scale is logarithmic, corresponding to the magnitude of $\lambda_1$. The plots reveal several zones, which are regular and chaotic. Blue regions indicate regular dynamics, while the rest of the domain is responsible for the system’s chaotic behavior. White dots represent
		periodic solutions of the variable-length pendulum with certain ratios of frequencies.}
		\label{fig:ciecie_variable}
	\end{figure*}
	\begin{figure}[t]
	\centering
	\includegraphics[width=0.99\linewidth]{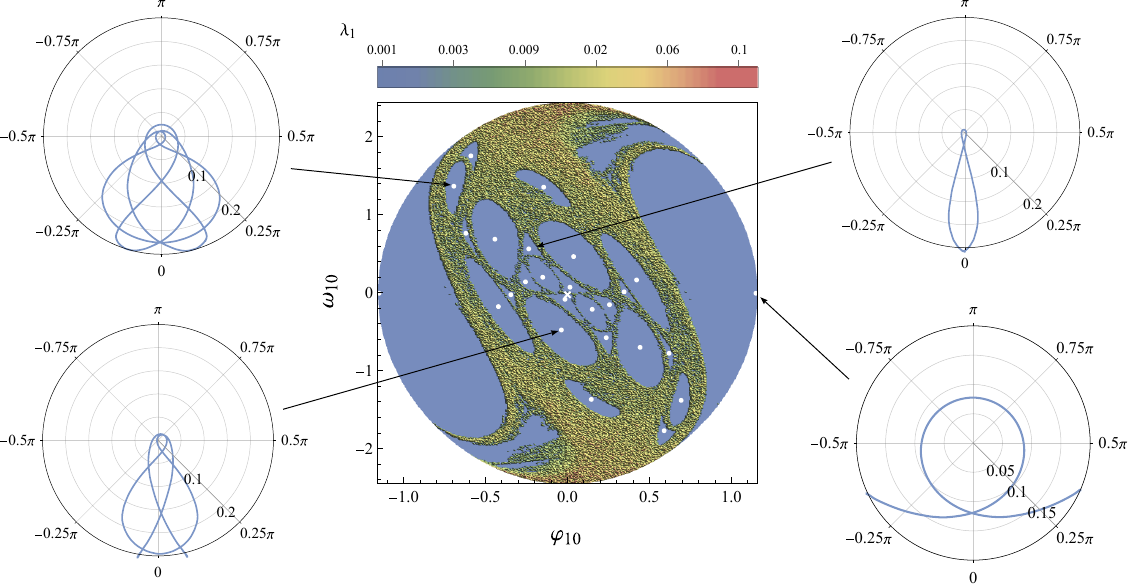}
	\caption{(Color online)  The  Lyapunov--Poincar\'e refined map together with polar plots computed for distinct values of the initial conditions marked at the maps showing different structures of the trajectories and their periodicity. In the radial direction of the polar plots, we measure $\ell(t)$; in the angular direction, we quantify $\varphi_1(t)$.\label{fig:lyap_pik}}
\end{figure}
\subsection{The Lyapunov--Poincar\'e refined map}

\begin{figure}[t]
		\centering
		\subfigure[Energy $E=0.5$, the coexistence of regular and chaotic domains with depicted periodic solutions]{
		\includegraphics[width=0.8\linewidth]{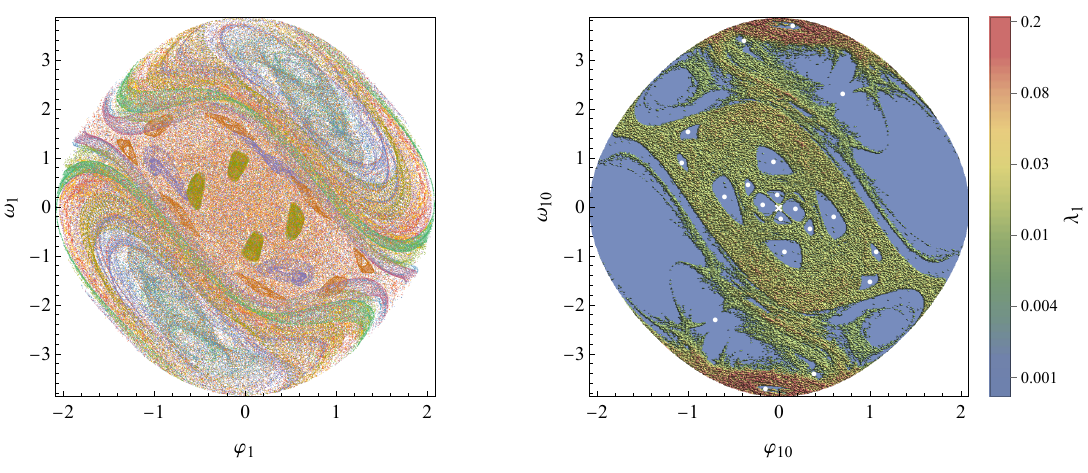}}\\
		\subfigure[Energy $E=0.65$, chaos occupies larger area of the Poincar\'e plane]{\includegraphics[width=0.8\linewidth]{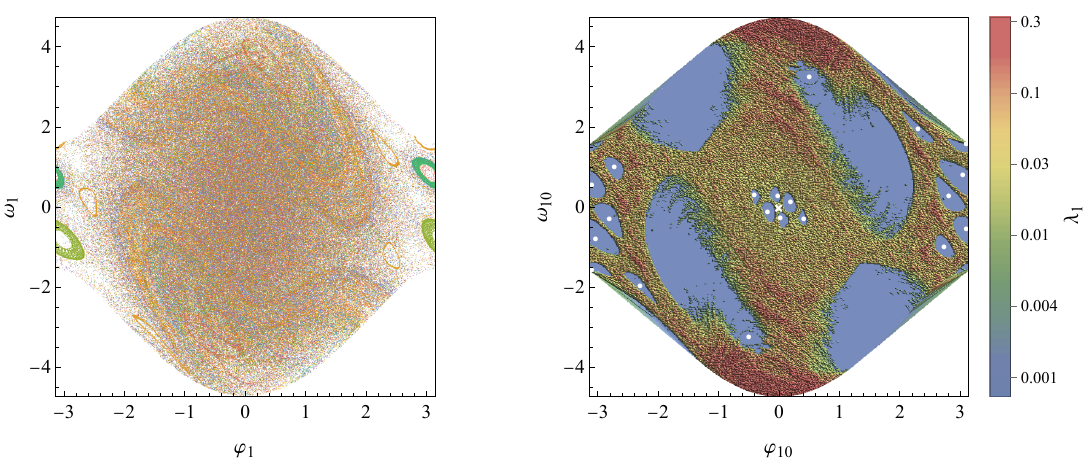}}
			\caption{(Color online) The Poincar\'e section of the system~\eqref{eq:rhs} and their corresponding  Lyapunov--Poincar\'e refined maps, constructed for constant values of the parameters $\mu_1=3$ and $\mu_2=0.2$, with varying energy levels. The cross-section was specified as $\ell=0.2$, with the direction $v>0$. In the Lyapunov--Poincar\'e refined map, the color scale is logarithmic, corresponding to the magnitude of $\lambda_1$. The plots reveal several zones, which are regular and chaotic. Blue regions indicate regular dynamics, while the rest of the domain is responsible for the system’s chaotic behavior. White dots represent
		periodic solutions of the variable-length pendulum with certain ratios of frequencies. }
		\label{fig:ciecie_variable2}
	\end{figure}
The computations of the Poincar\'e section of the classical double pendulum have been well established over the years. Indeed, in the case of Hamiltonian systems of two degrees of freedom, this method gives a quick insight into the dynamics of a considered model. Especially in the pendulum-like systems show the beautiful coexistence of periodic, quasi-periodic, and chaotic solutions within the Poincar\'e cross-section plane. In the context of the integrability of Hamiltonian systems, the Poincar\'e sections are very often used to confirm the chaotic nature of the studied system as well as an indicator of the possible presence of the first integrals (conservation laws) and integrable dynamics~\cite{Szuminski:23::,Szuminski:24::}. For a Hamiltonian system of two degrees of freedom system, the method of the Poincar\'e section typically results in a two-dimensional plot. If the system is integrable, the trajectories in phase space are confined to a surface of invariant tori, and the Poincar\'e section will show regular patterns with periodic solutions bounded by quasi-periodic loops. If the system exhibits chaotic behavior, the Poincar\'e section plane is covered by scattered, random-looking points corresponding to the fact that the trajectories can freely wander over large regions of the phase space. However, in the context of Hamiltonian systems of more than two degrees of freedom, this method is less practical for computational and technical reasons. 
The phase space of Hamiltonian systems with three degrees of freedom is richer and more complicated since trajectories lay in 3-dimensional tori, and their intersections with the Poincar\'e section will be two-dimensional surfaces rather than simple curves.
 Visualizing these sections is more complex than in the 2DoF case, and advanced techniques are needed to project or represent the dynamics effectively. For this purpose, we combine the Lyapunov exponents with the phase-parametric diagrams, which provide an excellent global view of the system's dynamics.

In Fig.~\ref{fig:ciecie_variable_0}(a), we can see the first Poincar\'e cross-section of the system created for $\mu_2=2\cdot 10^{-4}$ at the energy level $E=E_\text{min}+0.12$, where $E_\text{min}$ is the energy minimum at $\ell_0=0.2$ with zero values of the remaining  state variables.  The patterns shown in Fig.~\ref{fig:ciecie_variable_0}(a) are formed by tracing the intersections of numerically calculated orbits with a properly generic surface of section $\ell=0.2$, where $v > 0$, in a five-dimensional hypersurface defined by a constant energy level $E$, as defined in~\eqref{eq:energy}. These cross-section points were then projected onto $(\varphi_1,\omega_1,\varphi_2)$-space, and to the plane $(\varphi_1,\omega_1)$. From Fig.~\ref{fig:ciecie_variable_0}(a), it is evident that the system appears to be integrable. The Poincar\'e section displays regular patterns with two distinct particular periodic solutions enclosed by quas-periodic loops. Each loop represents a distinct initial condition. This depiction closely resembles the Poincar\'e section for the standard Hamiltonian system of two degrees of freedom, as for $\mu_2=2\cdot 10^{-4}$, our model can be treated as the integrable swinging Atwood machine perturbed by the second pendulum of negligible mass.  However, the situation could be more complex when we increase the value of $\mu_2$, which provides a greater impact on the system's dynamics coming from the second pendulum. Our suspicion is confirmed. The next Poincar\'e section, as visible in Fig.~\ref{fig:ciecie_variable_0}(b), shows complex, non-distinguishable system dynamics with multiple overlapping trajectories. 
It is unclear whether we have periodic solutions in the Poincar\'e  plane. Even more, regular, quasi-periodic trajectories are not distinguishable from chaotic orbits. Therefore, as mentioned previously,   the Poincar\'e section for Hamiltonian systems of multiple degrees of freedom is almost unpractical. To overcome the above difficulties, we combine the method of Lyapunov exponents with the phase-parametric diagrams.

We process slightly differently as in constructing the Lyapunov refined map depicted in Fig.~\ref{fig:bif_pik}. Namely, we fix values of the parameters $\mu_1$ and $\mu_2$. Then, we set a family of the initial conditions
\begin{equation}
\label{eq:ini_point}
\ell_0=0.2,\quad v_0=v\left(E,\varphi_{10},\omega_{10}\right),\quad \varphi_{10}=\varphi_{10},\quad \omega_{10}=\omega_{10},\qquad \varphi_{20}= \omega_{20}=0, 
\end{equation}
where $v(E,\varphi_{10},\omega_{10})$ is the positive square root of the quadratic equation for $v$ imposed by the energy conservation $E=\const$, defined in~\eqref{eq:energy}. Thus, for the fixed value of the energy, we have restrictions for $(\varphi_{10},\omega_{10})$ since the square root for $v$ must be real. Thus, we build a grid $B$ of a large number of initial conditions uniformly distributed in the available area of the Poincar\'e plane $(\varphi_{10},\omega_{10})$, and we compute the Lyapunov exponents for each initial condition. As a result, we obtain the Lyapunov exponent diagram showcasing the areas of regular and chaotic behavior of the system. Moreover, to get periodic solutions and their numbers, we use the algorithm discussed in Sec.~\ref{sec:refined} while constructing the Lyapunov refined map depicted in Fig.~\ref{fig:bif_pik}. Namely, from the obtained Lyapunov exponent diagram, we exclude these pairs $(\varphi_{10},\omega_{10})$ from the grid $B$, for which $\lambda>0$. Then, from the reduced grid of $n$ initial conditions, we compute the Poincar\'e sections, and we collect the points $(\varphi_1,\omega_1)$ when the trajectory passes through
 the cross-section plane $\ell_0=0.2$ with positive direction $v>0$,  for a sufficiently large number of sections. As a result we obtain $n$-lists containing sections points $(\varphi_1 ,\omega_1)$. Then, we perform the similarity check. 
In each list, we identify the pattern of repeating values of $\varphi_1$ and $\omega_1$ in a specific order. A finite set of repeating points is observed if the motion is periodic. We select the points closest to the periodic solutions to avoid redundant repetitions of the same families of solutions. In this way, we obtain a grid of the initial values $(\varphi_{10},\omega_{10})$, for which motions are periodic or very close to periodic ones. Combining this grid with the Lyapunov exponents diagram, we get a clear map showcasing the coexistence of regular and chaotic dynamics with depicted periodic solutions. We called such a map the \textit{Lyapunov--Poincar\'e refined map}.

The first two results of the above computations are depicted in Fig.~\ref{fig:ciecie_variable}, which show the classical Poincar\'e sections and their corresponding Lyapunov--Poincar\'e refined map.   From these maps, one can easily distinguish between the areas on the plane where the motion is regular or chaotic, which was almost impossible in the standard noised Poincar\'e section. 
Moreover, by the white dots, we identify periodic solutions with certain ratios of frequencies. The white cross located at the center represents a particular aperiodic solution lying on an invariant manifold, where the equations of motion~\eqref{eq:rhs}, when restricted, reduce to a system with one degree of freedom corresponding to the ordinary Atwood machine. Along this particular solution, we will perform an effective integrability analysis, which will be detailed in the next section.
As we expected for the infinitesimal value of the parameter $\mu_2=2\cdot 10^{-4}$, the Lyapunov--Poincar\'e refined map confirms the regular dynamics of the system as evidenced in Fig.~\ref{fig:ciecie_variable}(a). 
The image shows a regular pattern with $\lambda_1\approx 0$, for each pair $(\varphi_{10},\omega_{10})$ in the available area of the plane,  precluding the system's chaotic dynamics in this scenario. We have only two distinct particular periodic solutions marked by white dots.  Overall, we may conclude that for such a small value of $\mu_2$, the perturbation coming from the second pendulum is negligible, and the system's motion is almost integrable.  
 However, as $\mu_2$ increases, for instance, up to value $\mu_2=0.2$, the chaos appears, and we observe the classical coexistence of the regular domain bounded by necklace formations
responsible for the chaotic motion. Moreover, at the 
Lyapunov--Poincar\'e refined map one can also detect several families of periodic or very close to periodic one solutions of the variable-length pendulum. The polar plots of the mentioned periodic trajectories with the corresponding initial conditions depicted at the Lyapunov--Poincar\'e refined map are visible in Fig.~\ref{fig:lyap_pik}. As we can notice, we have certain families of periodic solutions with different shapes and periodicity. Overall, the Poincaré–Lyapunov refined map, within the framework of higher-dimensional Hamiltonian systems, can be regarded as a powerful tool, effectively replacing or complementing existing methods.

Usually, the Poincar\'e cross sections are constructed for fixed values of the parameters describing the system with varying energy. In this way, the evolution of the system's dynamics and the route to the chaos is given. Let us now consider  $\mu_1=3$ and $\mu_2=0.2$, as in Fig.\ref{fig:ciecie_variable}(b), with further 
increasing energy values. Fig.~\ref{fig:ciecie_variable2}  show the Poincar\'e cross sections and their corresponding Lyapunov--Poincar\'e refined maps for larger values of  $E$. We observe the classical path to global chaos. 
As the energy value increases, the area occupied by the chaotic motion of the system also expands, as expected. In fact, the refined Lyapunov–Poincar\'e maps, as shown in Fig.~\ref{fig:ciecie_variable}, exhibit structures very similar to the classical Poincar\'e cross-sections of the 2D double pendulum or the swinging Atwood machine model.   
However, it is worth mentioning that despite the large energy values, where the chaos dominates a significant portion of the admissible area, several families of periodic solutions with certain frequency rations can still be found.

All in all, we may conclude that the variable-length double pendulum, governed by the equations of motion~\eqref{eq:rhs} and the energy first integral~\eqref{eq:energy}, is a highly not integrable Hamiltonian system. We prove this fact in the following section.

%\begin{figure*}[htp]
%%
%		\centering
%		\includegraphics[width=0.45\linewidth]{mpoinlau_mu1_mu2_shortL1}\hspace{0.2cm}
%		\includegraphics[width=0.45\linewidth]{mpoinlau_mu1_mu2_shortL2}	\caption{$\mu_1=3,\mu_2=0.2, \ell_0=0.2, v_\ell_0=0, \omega_{10}=0, \omega_{20}=0$}
%		\label{fig:parki_min_max}
%	\end{figure*}

	\section{Integrability \label{sec:integrability}}
	As we have seen, the numerical analysis shows the system's complicated and hyperchaotic dynamics. Nevertheless, the possible existence of additional first integrals for specific values of the parameters is still unknown. Therefore, a comprehensive integrability analysis is needed. 
	We state the following theorem.
	\begin{theorem}
	\label{th:int_{2}}
	Let $\mu_{1}$ and $\mu_{2}$ are positive, non-zero parameters. If the variable-length double pendulum governed by equations of the motion~\eqref{eq:rhs}  is integrable in the sense of Liouville, then 
	\begin{equation}
	\label{eq:waunrek}
	\mu_{1}=\frac{M}{m_{1}+m_{2}}= 1+\frac{4}{n(n+3)-2},\quad n\in \N_{>0}.
	\end{equation}
	\end{theorem} 
\begin{proof}
	To simplify the further analysis of variational equations and their differential Galois groups, we introduce the following positive  parameters
	\begin{equation}
		M_1\equiv\frac{\mu_1-1 }{\mu_1+1},\quad M_2\equiv\frac{2\mu_1}{(\mu_1+1)(1-\mu_2)}.
	\end{equation}
	From~\eqref{eq:parki1}, the parameters $M_1, M_2$ have the following restrictions
	\begin{equation}
		\label{eq:restric}
		0<M_1<1,\quad\text{and}\quad M_2\geq 1+M_1.
	\end{equation}
	System~\eqref{eq:rhs} has the following invariant manifold
	\begin{equation*}
		\label{eq:scn}
		\scN=\left\{(\ell,v,\varphi_1,\omega_1,\varphi_2,\omega_2)\in \field{C}^6\, \big{|}\, \varphi_1=\varphi_2=0=\omega_1=\omega_2\right\}.
	\end{equation*}
	Restricting the right-hand sides of~\eqref{eq:rhs} to $\scN$, we obtain one degree of freedom
	system  
	\begin{equation}
		\dot \ell=v,\quad \dot v=-M_1.
	\end{equation}
	%  with the energy first integral
	%  \begin{equation}
		% \mathcal{E}=\left(\frac{\mu_1+\mu_2}{2}\right)v^2+(\mu_1-\mu_2)\ell.
		%  \end{equation}
	The above equations describe the motion of the ordinary Atwood machine with the following solutions
	\begin{equation}
		\label{eq:ell-vl}
		\ell(t)=-\frac{1}{2}M_1^2t^2+v_0t+\ell_0,\quad v(t)=-M_1t+v_0.
	\end{equation}
	Here $\ell_0$ and  $v_0$ are respective constants of the integrations related to an initial length  $\ell_0=\ell_0$ and initial velocity $v_0=v_0$. 
    For further considerations, we fix $v_0=0$, and $\ell_0=M_1/2$. 
 Thus, a particular solution of system~\eqref{eq:rhs} is $\vvarphi(t)=(\ell(t),v(t),0,0,0,0)$, along which we compute the variational equations.

	Let $\vX=[L,V,\Theta_1,\Omega_1,\Theta_2,\Omega_2]^T$ denote the variations of $\vx=[\ell,v,\varphi_1,\omega_1,\varphi_2,\omega_2]^T$. Then, the variational equations of system~\eqref{eq:rhs} along  $\vvarphi(t)$, are as follows
	\begin{equation}
		\dot \vX=\vA(t)\cdot \vX,\quad \text{where}\quad \vA=\frac{\partial \vv}{\partial\vx}(\vvarphi(t)).
	\end{equation}
	Here $\vv$ denotes the right-hand side of system~\eqref{eq:rhs}. The explicit form of the non-constant  matrix $\vA$ is given by
	%  \begin{equation*}\vA=
		% \begin{pmatrix}
			% 0&0&0&1&0&0\\[.15cm]
			% 0&0&0&0&1&0\\[.15cm]
			% 0&0&0&0&0&1\\[.15cm]
			% 0&0&0&0&0&0\\[.15cm]
			% 0& \frac{\mu_1-\mu_2+\mu_1\mu_2+\mu_2^2}{(1-\mu_2)(\mu_1+\mu_2)\ell}& \frac{2\mu_1}{(\mu_2-1)(\mu_1+\mu_2)\ell}&0&-\frac{2v_\ell  }{\ell}&0\\[.15cm]
			% 0&\frac{2\mu_1\mu_2}{(\mu_2-1)(\mu_1+\mu_2)}&\frac{2\mu_1\mu_2}{(1-\mu_2)(\mu_1+\mu_2)}&0&0&0\\
			% \end{pmatrix}
		%  \end{equation*}
	\begin{equation*}\vA=
		\begin{pmatrix}
			0&1&0&0&0&0\\
			0&0&0&0&0&0\\
			0&0&0&1&0&0\\
			0&0&\frac{M_1-M_2}{\ell}&-\frac{2	v  }{\ell}&\frac{M_2-M_1-1}{\ell}&0\\[.15cm]
			0&0&0&0&0&1\\[.15cm]
			0&0&M_2&0&-M_2&0\\
		\end{pmatrix}.
	\end{equation*}
	We notice that variational equations separate into
blocks: normal variational equations for variables $[\Theta_1,\Omega_1,\Theta_2,\Omega_2]^T$ and tangential equations for variables $[L,V]^T$. Because the tangential system is trivially solvable, for further integrability analysis, we take the normal part. It has the form
\begin{equation}
	\label{eq:normal}
	\begin{split}
   \ddot \Theta_1+\left(\frac{2v}{\ell}\right)\dot\Theta_1+\left(\frac{M_2-M_1}{\ell}\right)\Theta_1+\left(\frac{M_1-M_2+1}{\ell}\right)\Theta_2=0,\qquad \ddot\Theta_2+M_2(\Theta_2-\Theta_1)=0.
    \end{split}
\end{equation}
We introduce a new parameter $n$ defined by relation 
\begin{equation}
\label{eq:M1}
    M_1=\frac{2}{(3 + n) n},  \quad n>0, 
\end{equation}
 and we rewrite system~\eqref{eq:normal} as  one  fourth-order differential equation for the variable $y\equiv \Theta_2$:  	% \begin{equation}
	% 	\label{eq:fourth-order}
	% 	\begin{split}
	% 		\ddddot \Theta_2+2\left(\frac{v}{\ell}\right)\dddot \Theta_2+\left(M_2+\frac{M_2-M_1}{\ell}\right)\ddot\Theta_2+2M_2 \left(\frac{v}{\ell}\right)\dot\Theta_2+\left(\frac{M_2}{\ell}\right)\Theta_2=0.
	% 	\end{split}
	% \end{equation}
\begin{equation}
		\label{eq:fourth-order2}
		\begin{split}
			  (t^2-1)\ddddot y +4t\dddot y +[2-M_2(1+3n+n^2-t^2)]\ddot y  + 4M_2t\dot y-n(3+n) M_2y=0.
		\end{split}
	\end{equation}
	The non-integrability proof of the system relies on the necessary conditions outlined in Lemma~\ref{lem:5}, which is presented in the Appendix. To begin, we must verify whether equation~\eqref{eq:fourth-order2} has a hyperexponential solution. A function  $f(z)$  is defined as hyperexponential if its logarithmic derivative,  $f'(z)/f(z)$, is a rational function. The following lemma serves as the basis for determining the existence of a hyperexponential solution and, consequently, the potential integrability of the system.
 \begin{lemma}\label{eq:lemma2}
     If $n>0$ is not an integer, then Eq.~\eqref{eq:fourth-order2} does not admit a hyperexponential solution. 
 \end{lemma}
 \begin{proof}  Eq.~\eqref{eq:fourth-order2} has two regular singular points $t_{1}=1$ and $t_2=-1$, and one irregular at infinity $t_\infty=\infty$. 
 Exponents at these points belong to the following sets
 \begin{equation}
 \label{eq:exp}
 E_1=\{ 0, 1, 2\},\quad E_2=\{ 0, 1, 2\},\quad E_\infty=\{-n,n+3,\pm \mathrm{i}\sqrt{M_2}t\}.
 \end{equation}
First, we examine whether the equation allows for rational solutions, which are a special type of hyperexponential solution. If such a solution exists, it will take the form
\begin{equation}
\label{eq:sol_rat}
y(t)=P(t)(t-t_1)^{e_1}(t-t_2)^{e_2}, \quad \text{where}\quad e_1,e_2\in \{ 0, 1, 2\}.
\end{equation}
Here $P(t)$ is a polynomial of degree $d= -e_1-e_2 - e_\infty$, where $ e_\infty \in \{-n,n+3\}$.
Hence, if a rational solution exists, then it is polynomial.  As $e_1\geq
0$ , and $e_2\geq 0$, there is only one choice $e_\infty=-n$. Then by a direct
check one can verify that for a given integer $n>0$,
equation~\eqref{eq:fourth-order2} has a polynomial solution of degree $n$.

To check whether the equation admits a hyperexponential solution, it is
convenient to move the irregular point to the origin. So, we change the
independent variable $t=1/\tau$. Now, the transformed equation has a regular
singular points at $\tau_{1}=1$ and $\tau_2=-1$, with the same exponents as
before the transformation. For the irregular singularity at $\tau_3=0$, the
generalised exponents are $E_\infty=\{-n,n+3,\pm \mathrm{i}\sqrt{M_2}/\tau\}$.
We know that if a hyperexponential solution exists, then it is of the form 
\begin{equation}
	\label{eq:hyexp2}
	y(\tau) = \scE(\tau) R(\tau)(\tau-\tau_1)^{e_1}(\tau-\tau_2)^{e_2},
\end{equation}
where $R(\tau)$ is a rational function and $\scE(\tau)$ is called the exponential part of formal solution. 
Using Maple\footnote{To be precise, we use Maple package ''DEtols'' with ''formal\textunderscore sol'' procedure. }, we found that the exponential parts of formal solutions of Eq.~\eqref{eq:fourth-order2} are the following
\begin{equation}
     	\scE(\tau)	\in \{\tau^{-n},\tau^{n+3},\exp[\pm\mathrm{i}\sqrt{M_2}/\tau]\}.
 \end{equation}
Let us assume that $n$ is not a positive integer. Then,  the first two possible choices for $ \scE(\tau)$ are
excluded. Therefore, if there exists a hyperexponential solution, it is of the
form $y(\tau)= \exp[\pm\mathrm{i}\sqrt{M_2}/\tau]R(\tau)$. With this
substitution, we obtain the equation for $R(\tau)$. The necessary condition for the existence of a rational solution is that there are exponents $e_1,e_2$, and $e_\infty$, such that their sum is a non-positive integer. But by assumption, $n$ is not an integer, so it is impossible. 
 This ends the proof.
 \end{proof}
 By the main Lemma~\ref{lem:5} given in the Appendix, we have shown that for $M_{1}$ given in~\eqref{eq:M1} with a noninteger $n$ the normal variational equations~\eqref{eq:normal} do not possess any hyperexponential solution and to complete the non-integrability proof, we have to show that its corresponding external second power has exactly one solution. 
 \begin{lemma}\label{eq:lemma2ex}
     If $n>0$ is not an integer, then the second exterior power of equation~\eqref{eq:fourth-order2} has exactly one hyperexponential solution. 
 \end{lemma}
 \begin{proof}
 To verify the second assumption of Lemma~\ref{lem:5}, we need to examine the second exterior power of equation~\eqref{eq:fourth-order2}, which corresponds to the system in~\eqref{eq:7b}. When expressed as a scalar equation, this system becomes a sixth-order equation, given by
     \begin{equation}
	\label{eq:ex22}
	\sum_{i=0}^6 b_i(\tau) y^{(i)} = 0, \qquad b_i(\tau)\in \Q(\tau),
\end{equation}
 where we introduced a new independent variable $\tau=1/z$. The rational coefficients $b_i(\tau)$ are rather long, so we do not write them here explicitly.
 Eq.~\eqref{eq:ex22} has three regular singularities $\tau_1=-1$, $\tau_2=1$, and $\tau_3=\infty$ with the respective sets of exponents: 
\begin{equation}
    E_1 = \{0,1,2,3\},\quad E_2 = \{0,1,2,3\}, \quad E_3=\{ 2,3,4,5,6,8\}
\end{equation}
At $\tau = 0$ there is an irregular singular point. Using Maple, we found that the exponential parts  of a formal solution $y(\tau)$ 
of~\eqref{eq:ex22} are  
\begin{equation}
    \scE(\tau)\in\left\{ \tau^{\pm 2}, \tau^{n+1} \rme^{\pm u/\tau}, \tau^{-(n+2)} \rme^{\pm u/\tau} \right\}.
\end{equation}
One can check that there are only two exponents $e_0=\{\pm2\}$. Now, it is easy to check that equation~\eqref{eq:fourth-order2} has only one rational solution $y(\tau) = \tau^{-2}$. Indeed, if a rational solution exists, then there exist integer exponents $e_i\in E_i$ such that $e_0+e_1+e_2+e_3$ is a nonpositive integer. There is only one such choice, namely $ e_0=-2$, $e_1=e_2=0$ and $e_3=2$.
A possible hyperexponential solution of~\eqref{eq:ex22} has the form 
\begin{equation}
	\label{eq:hyexp22}
	y(\tau) = \scE(\tau) R(\tau)(\tau-\tau_1)^{e_1}(\tau-\tau_2)^{e_2},
\end{equation}
where $R(\tau)$ is a rational function and $e_i\in E_i$. Taking into account the possible values of $e_1$ and $e_2$, we see that the possible hyperexponential solution
is of the form $ y(\tau) = \scE(\tau) R(\tau) $.  If $\scE(\tau)=\tau^{\pm2}$, then $y(\tau)$ is rational, but we know that there is only one rational solution. In effect, we have to investigate four cases $\scE(\tau)\in \left\{  \tau^{n+1} \rme^{\pm u/\tau}, \tau^{-(n+2)} \rme^{\pm u/\tau} \right\}$.

Let us consider the first possibility. Inserting $y=R(\tau)\tau^{n+1} \rme^{\pm u/\tau} $
into~\eqref{eq:ex22}, we obtain the equation of sixth order for the rational function $R(\tau)$. It has the same sigularities as equation~\eqref{eq:ex22}. The sets of exponents $E_1$ and $E_2$ do not change, while   
\begin{equation*}
  E_0 =\{0,-3-2n
\}\cap\Z=\{0\},\qquad   E_3 = \{ 3+n,n+4,n+5,n+6,n+7,n+9\}\cap\Z = \emptyset,
\end{equation*}
because by assumption $n\not\in\Z$. Hence, there is no rational solution for $R(\tau)$. 
In a similar way we exclude possiblity that $y=R(\tau)\tau^{-n-2} \rme^{\pm u/\tau} $. 
We conclude that if $n$ is not a positive integer, then the second exterior power of equation~\eqref{eq:fourth-order2} has only one hyperexponential solution, which ends the proof.
 \end{proof}
 		\begin{figure*}[t]
		\centering
			\includegraphics[width=0.33\linewidth]{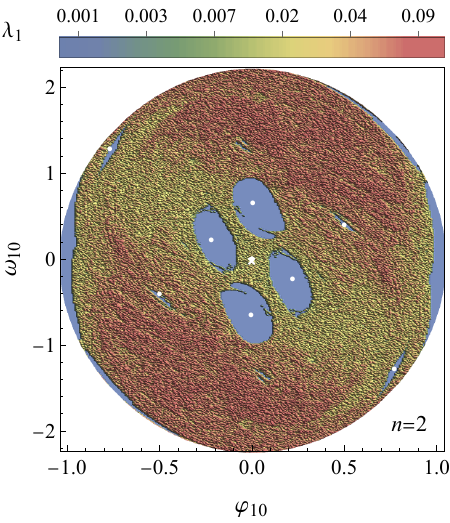}		\includegraphics[width=0.33\linewidth]{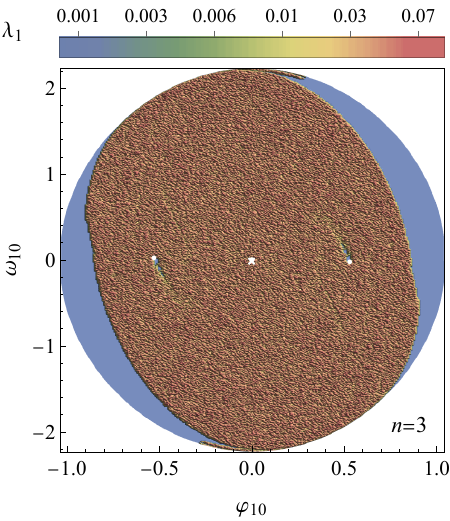}	\includegraphics[width=0.33\linewidth]{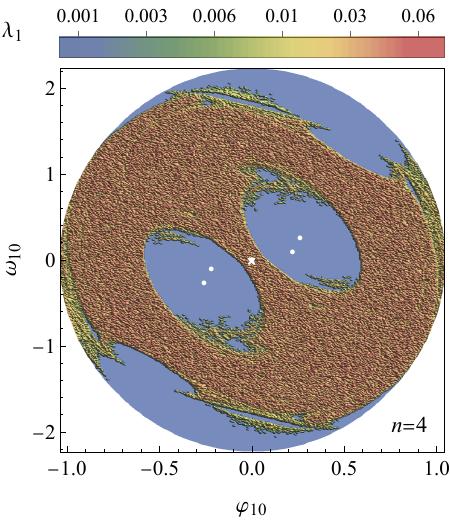}\\
						\caption{(Color online) The Lyapunov--Poincar\'e refined maps of the system~\eqref{eq:rhs}, constructed for constant values of the parameters $\mu_1=1+4/[n(n+3)-2]$ and $\mu_2=0.2$ with varying $n$,  at the constant energy level $E=E_\text{min}+0.1$.  Here, $E_\text{min}$ is the energy minimum at $\ell_0=0.2$ with zero values of the remaining state variables.   The color scale is logarithmic, corresponding to the magnitude of $\lambda_1$. The plots reveal several zones, which are regular and chaotic. Blue regions indicate regular dynamics, while the rest of the domain is responsible for the system’s chaotic behavior. White dots represent
		periodic solutions of the variable-length pendulum with certain ratios of frequencies.}
		\label{fig:ciecie_variable_rest}
	\end{figure*}
We have proved that if $n$ is not a positive
integer, the fourth-order variational equation~\eqref{eq:fourth-order2} does not
possess any hyperexponential solution. Thus, according to Lemma~\ref{lem:5}, the system
is not integrable, in the sense of Liouville.
\end{proof}
The  question about possible integrable cases of the system for $\mu_1$ satysfaing~\eqref{eq:waunrek} remains open. We already know due to the detailed numerical analysis, that for $n=1$, which implies $\mu_1=3$, the dynamics of the system is highly chaotic, precluding its integrability. However, the integrability for larger $n\in \N_{>1}$, needs the additional view. Fig.~\ref{fig:ciecie_variable_rest} illustrated three Poincar\'e--Lyapunov refined maps, constructed for $n=2,3,4$ under the initial conditions~\eqref{eq:ini_point}, at the constant energy levels   $E=E_\text{min}+0.1$. As we can notice, the chaos occupies most available areas on the maps. One can find only a few periodic solutions bounded by quasi-periodic or chaotic trajectories. 

With the help of Lyapunov exponents, we can easily estimate the percentage of chaos,
 as well as the maximum value of the largest Lyapunov exponent, equipped in the available area of the  Lyapunov–Poincar\'e refined maps  as
 a function of $n$. 
 Results for $n=5,\ldots, 50$ are shown in Fig.~\ref{fig:percentege}. As we can see,  the percentage of chaos increases rapidly as $n$ is increasing as well.  Indeed, over the range $n\in \{14\ldots,20\}$, the system goes to an almost fully ergodic one with the percentage of chaos $~ 99\%$. Then, for $n>20$, the situation changes.  The system's dynamics start to be more regular, and the percentage of chaos is monotonically decreasing and for  $n\geq 46$, the chaos is negligible. 
This decline in chaos for larger  $n$  values occurs because the mass ratio  $\mu_1$  approaches the critical threshold  $\mu_1 = 1$, where the pendulums undergo a transition to a state where they effectively slide down. Additionally, the maximum values of  $\lambda_1$, the largest Lyapunov exponent, generally decrease as a function of  $n$, as expected. 
From the above, we formulate the following conjecture.

\begin{conjecture}
Let  $\mu_{2}$ be a non-zero, positive parameter and $\mu_{1}$ provided by~\eqref{eq:waunrek}. Then, the variable-length double pendulum   governed by equations of the motion~\eqref{eq:rhs},  is not integrable in the sense of Liouville.
\end{conjecture}
%We do not yet know how to prove this conjecture for arbitrary  $n$. However, in the following subsection, we demonstrate that the conjecture holds for the case  $n = 1$, which implies  $\mu_1 = 3$.
\subsection{The reduction of the variational equations}

				\begin{figure*}[t]
		\centering 
\includegraphics[width=0.58\linewidth]{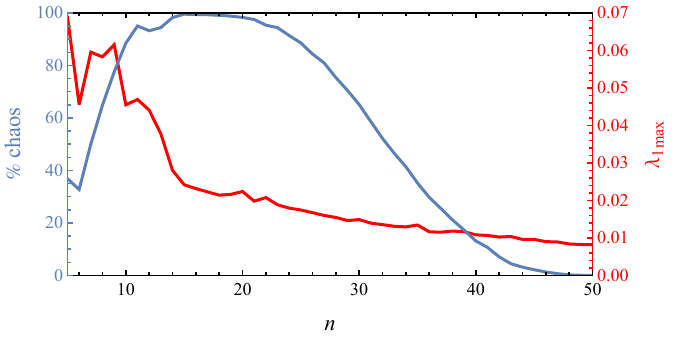}
\caption{(Color online) The percentage of chaos against the maximal values of the largest Lyapunov exponent $\lambda_1$ in phase space as a function of the parameter $n\in \{5,\ldots, 50\}$. }
\label{fig:percentege}
	\end{figure*}
For fixed $n\in \N$ the fourth-order variational
equation~\eqref{eq:fourth-order2} factors in as $L_1[L_2[L_3y(t)]]=0$, where
$L_1,L_3$ are differential operators of order one, while $L_2$ has order
two. The coefficients of the operators $L_i$ are unknown rational functions of $t$
whose degrees depend on $n$. Unfortunately, we are not able to write general
formulae for operators $L_i$ keeping $n$-arbitrary. However, the right
operator $L_3$ has order one. Thus, the variational
equation~\eqref{eq:fourth-order2} has a hyperexponential
solution~\eqref{eq:sol}. Because $e_{1},e_2\in \{0,1,2\}$, the
solution~\eqref{eq:sol} translates into the following form.
\begin{equation}
    y(t)=P(t),\quad \quad P(t)\in C[t],\quad 
		\quad \operatorname{deg}P(t)=n.
\end{equation}
Then, substituting $y(t)=P\int x(t)$ we reduce the order of variational
equation~\eqref{eq:fourth-order2},  where $x(t)$ is a new dependent variable.
Next, thanks to the presence of the first integral of the reduced third-order
variational equation, further reduction is possible. As a result, we obtain a
second-order differential equation for which the exhaustive integrability
analysis is possible. Let us show the above procedure in an example of $n=1$.  Hence, let us prove the following
theorem.
\begin{theorem}
For $\mu_1=3$ and $0<\mu_2<	1$, the variable-length double pendulum defined by equations of the motion~\eqref{eq:rhs}  is not integrable in the class of functions meromorphic in coordinates and velocities.
\end{theorem}
\begin{proof}
For $n=1$,  the variational  equation~\eqref{eq:fourth-order2} takes the form
\begin{equation}
\label{eq:fourth_m1}
(t^2-1)\ddddot y+4t \dddot y+\left[2+M_2(t^2-5)\right]\ddot y+4M_2  t\dot y-4M_2y=0.
\end{equation}
It is easy to verify that  
$y(t)=t$ is the solution of~\eqref{eq:fourth_m1}. Thus, making the change of the dependent variable $y(t)=t\int x(t)$, we reduce equation~\eqref{eq:fourth_m1} to a third-order differential equation. It has the form
\begin{equation}
\label{eq:third_order}
\begin{split}
     t(t^2-1)\dddot x+(8t^2-4)\ddot x+\left[14+M_2(t^2-5)]t\, \dot x +[4+2M_2(3t^2-5)\right]x=0.
    \end{split}
\end{equation}
Moreover,  reduced equation~\eqref{eq:third_order} can be further factorized thanks to the existence of  the first integral 
\begin{equation}
\begin{split}
F&=\left[(6+M_2(t^2-5))t^2(t^2-1)\right]\ddot x+2\left[12t^2-6+M_2(t^4-10t^2+5)\right]t\,\dot x\\ & +\left[M^2t^2(t^2-5)^2+12(t^2+1)+10M_2(t^4-4t^2-1)\right]x.
\end{split}
\end{equation}
Indeed, $\dot  F$ generates~\eqref{eq:third_order}.  
Therefore, it is sufficient to study the differential Galois group of the second-order differential equation $F=0$, which can be rewritten as
\begin{equation}
\label{eq:second_son}
    \ddot x+p(t)\dot x+q(t)x=0,
\end{equation}
where
\begin{equation}
\label{eq:ppqq}
\begin{split}
 &p(t)=\frac{2}{t}+\frac{1}{t-1}+\frac{1}{t+1}-\frac{2M_2t}{6+M_2(t^2-5)},\\ & q(t)=M_2-\frac{2}{t^2}+\frac{2(1-M_2)}{t-1}-\frac{2(1-M_2)}{t+1}+\frac{2M_2}{6+M_2(t^2-5)}.
 \end{split}
\end{equation}
To apply the Kovacic algorithm, we rewrite equation~\eqref{eq:second_son} in its reduced form
\begin{equation}
\label{eq:second_red}
    \ddot u=r(t) u,\qquad \text{where}\quad u=\left(\frac{t\sqrt{t^2-1}}{6+M_2(t^2-5)}\right)x,
\end{equation}
and the coefficient
\begin{equation}
    \label{eq:RR}
    \begin{split}
    r(t)=\frac{2}{t}-\frac{1}{4(t-t_1)^2}-\frac{1}{4(t-t_2)^2}+\frac{3(5M_2-6)}{M_2(t-t_3)^2(t-t_4)^2} +\frac{-18+M_2(79-23t^2-2M_2(t^2-5)^2}{2M_2(t-t_1)(t-t_2)(t-t_3)(t-t_4)}.
    \end{split}
\end{equation}
Equation~\eqref{eq:second_red} has six singularities
\begin{equation}
    t_0=0,\quad t_1=1,\quad t_2=-1,\quad t_{3,4}=\pm\sqrt{5-\frac{6}{M_2}},\quad t_\infty=\infty
\end{equation}
Points $t_i=\{t_0,t_1,t_2,t_3,t_4\}$ are regular singularities of order two, while $t_\infty$ is irregular with zero degree of infinity.
The respective differences of exponents at $t_i$ are given by
\begin{equation}
\Delta_0=3,\qquad\Delta_1=\Delta_2=0,\qquad \Delta_3=\Delta_4=2,
\end{equation}
where $\Delta_i=\sqrt{1+4c_i}$ and $c_i$ are coefficients of the dominant terms of the Laurent series expansions of $r(t)$  around $t_i$.
To avoid the confluences of the poles, we assume $M_2\notin\{3/2,6/5\}$. Nevertheless, these values of $M_2$
where previously excluded due to our physical assumption~\eqref{eq:restric}.

Now, we can prove the following.
\begin{lemma}
    The differential Galois group of reduced equation~\eqref{eq:second_red} is $\operatorname{SL}(2,\C)$.
\end{lemma}
\begin{proof}
    As $\Delta_{1,2}=0$ local solutions in a neighborhood of points $t_\star=t_1$ and $t_\star=t_2$, contain a logarithmic term. Then, according to the well-known theory, two linearly independent solutions $u_1$ and $u_2$ of equation~\eqref{eq:second_red} in a vicinity of $t_\star\in\{t_1,t_2\}$, have the following forms
    \begin{equation}
       u_1(t)= (t-t_\star)^\rho f(t)\quad u_2(t)=z_1(t)\ln(t-t_\star)+(t-t_\star)^\rho h(t),
    \end{equation}
    where $f(t)$ and $h(t)$ are analytic at $t_\star$ and $f(t_\star)\neq 0$. In the paper~\cite{Maciejewski:02::}, authors proved that when at least one of the differences of the exponents is zero, then a continuation of the matrix of the fundamental solutions along a small lop encircling $t_\star$ counterclockwise gives rise to a triangular non-diagonal monodromy matrix. Thus, the differential Galois group of~\eqref{eq:second_son} cannot be finite nor dihedral. Therefore, $\mathcal{G}$ must be either triangularizable or the full group $\operatorname{SL}(2,\mathbb{C}) $. 
If the differential Galois group of the reduced second-order differential equation~\eqref{eq:second_son} with $r(t)\in \Q(t)$ is triangulisable, then it possesses a non-zero solution of the form
\begin{equation}
\label{eq:sol}
u(t)=P\exp\left[\int \omega\right],\quad \text{where}\quad P\in \R[t],\quad \omega\in \Q(t).
\end{equation}
If such a solution exists, it can be identified using the first case of the Kovacic algorithm. However, if the algorithm fails at any step, this indicates that no such solution exists, and consequently, $\mathcal{G}$ is not triangularizable.

We proceed with the first case of the Kovacic algorithm. For each singular point  $t_i$, we identify the corresponding sets of exponents.
\begin{equation}
    E_i:=\frac{1}{2}\left\{1\pm \Delta_i\right\},\quad \text{for}\quad i=0,\ldots, 4.
\end{equation}
Because the degree of infinity is $v=0$, an auxiliary set $E_\infty$ has a form $E_\infty=\{\pm b/2a\}$, where $a$ is a free term in the expansion $[\sqrt{r}]_\infty$, while $b$ is the coefficient in $t^{-1}$ of $r$ minus the coefficient of $t^{-1}$ in $([\sqrt{r}]_\infty)^2$. In our case $a=\rmi M_2$ and $b=0$. Therefore, the explicit forms of the auxiliary sets $E_i$ and $E_\infty$ are given by
\begin{equation}
\begin{split}
 E_1=\{2,-1\},\quad E_2=E_3=\{1/2,1/2\},\quad E_4=E_5=\{3/2,-1/2\},\quad E_\infty=\{0,0\}.
    \end{split}
\end{equation}
Next, from the Cartesian product $E=\prod_{i=0}^4E_i\times E_\infty$,  we select those elements $e=(e_0,e_1,e_2,e_3,e_4,e_\infty)$ for which
\begin{equation}
    d(e)=e_\infty-\sum_{i=0}^4e_i\in \N_0.
\end{equation}
In our case, however, no $d(e)$ satisfies this condition. Consequently, its differential Galois group cannot be conjugated to any triangular subgroup of $ \operatorname{SL}(2, \mathbb{C}) $. Since all cases of the Kovacic algorithm fail, we conclude that  $\scG = \operatorname{SL}(2, \mathbb{C}) $, completing the proof.\end{proof}

Our analysis leads us to conclude that the differential Galois group of equation~\eqref{eq:second_red} is $ \operatorname{SL}(2, \mathbb{C})$ with a non-Abelian identity component. This result implies that the variable-length double pendulum, described by system~\eqref{eq:rhs}, is not integrable in the sense of Liouville for the mass ratio  $\mu_1 = 3$  and any value of  $\mu_2$  within the range  $0 < \mu_2 < 1$.
\end{proof}
\section{Conclusions \label{sec:conclusions}}
The complicated and mostly chaotic dynamics of multiple pendulums is well known but still in great scientific activity. This is because such systems served as a fundamental framework within classical mechanics
and dynamical systems theory, providing insights into the coexistence of regular and chaotic motion. The systems of multiple pendulums
explain many fundamental phenomena and have found applications in physics, astronomy, engineering, and synchronization theory. Currently, we can observe increased work on the study of dynamics and chaos in the variable-length double pendulums. Compared to classic pendulums, variable-length double pendulums present a more complex dynamic behavior due to the time-dependent nature of their arm lengths. This variability introduces additional degrees of freedom to the system, making the traditional numerical/analytical techniques significantly more challenging. Specifically, our goal was to address the limitations of classical methods, such as Poincar\'e sections, when applied to systems with higher-dimensional phase spaces. These traditional techniques, while effective in low-dimensional systems, become less informative and harder to interpret in higher-dimensional contexts due to the complexity of the phase space. To overcome these difficulties, we propose a novel approach -- the Lyapunov refined maps. This new method combines the strengths of Poincar\'e sections, phase-parametric diagrams, and Lyapunov exponents into one unified framework. By doing so, it provided both qualitative and quantitative insights into the system’s dynamics. The Lyapunov refined maps allowed us to measure the strength of chaos while also capturing periodic, quasi-periodic, and chaotic behaviors in a single comprehensive view. This approach provides a stronger tool for studying complex systems, especially in situations where traditional methods have difficulty offering clear insights.

  The detailed analysis suggested the non-integrability of the proposed model.  We gave the analytical proof of this fact. Thanks to the presence of particular solutions, we were able to apply the Morales-Ramis theory.  The novelty of our work concerns the fact that we performed the integrability analysis of the Hamiltonian systems of three degrees of freedom for which the variational equations transform into fourth-order differential equations.  For the effective analysis of the differential Galois group of variational equations, we apply a recently formulated extension of the Kovacic algorithm to dimension four.   For some values of the parameters satisfying the necessary integrability conditions, we used the Lyapunov exponents diagram as an indicator of the presence of the first integrals and integrable dynamics. We have shown that for the mass ratio $\mu_1$, for which the necessary conditions of integrability are satisfied, the system exhibits strong chaotic behavior visible in the Lyapunov diagrams, which unequivocally excludes its integrability. 

 In summary, our article provides an in-depth analysis of the dynamics and integrability of the variable-length double pendulum, offering novel insights and methodologies for future research in this area. We utilized advanced tools to achieve results that are highly relevant and valuable for the study of differential equations, vibrations, and synchronization theory. Our results not only contribute to the fundamental understanding of variable-length pendulum dynamics but also highlight their potential applications in adaptive robotics, crane models, energy harvesting, and biomechanics.

\section*{Acknowledgements}
 For Open Access, the authors have applied a CC-BY public copyright license to any Author Accepted Manuscript (AAM) version arising from this submission. We would like to give special thanks to
A.~J. Maciejewski for the discussion and his insightful comments on this work. The present work was carried out during a research internship of W. Szumiński, conducted under the guidance of T. Kapitaniak.
		\section*{Funding} This research has been  founded by The
	National Science Center of Poland under Grant No.
	2020/39/D/ST1/01632 and  
by the Polish Ministry of Science program Regional Excellence Initiative,
RID/SP/0050/2024/1.
	\section*{Data availability }
	The data that support the findings of this study are available from the corresponding author,  upon reasonable request.
	\section*{Declarations}
	\textbf{Compliance with ethical standards}\\
	\textbf{Conflicts of interest} The authors declare no conflict of interest.\\

\appendix
\section{Integrability: useful facts and statements}
During the applications of the Morales-Ramis theory, it is crucial to determine whether the differential Galois group of the variational equations is virtually Abelian—that is, whether the group’s identity component is Abelian. In most cases where Morales-Ramis theory has been successfully applied, the variational equations decompose into a system of second-order equations or contain at least one second-order equation as a subset. Additionally, it is often feasible to transform these second-order equations into ones with rational coefficients.
This transformation makes it possible to apply the Kovacic algorithm, as introduced in \cite{Kovacic:86::}. The algorithm tells whether the given second-order differential equation has closed-form solutions, specifically by identifying if solutions exist within the field of Liouvillian functions. The Kovacic algorithm operates by classifying all possible algebraic subgroups of $\mathrm{SL}(2,\mathbb{C})$, which enables it to characterize the differential Galois group of the equation as a by-product of this classification.

Let us consider a Hamiltonian system  with $n$ degrees of freedom:
\begin{equation}
\label{eq:rs0}
  \dot \vz=\mathbb{J} H'(\vz),\quad \mathbb{J}=\begin{bmatrix} 0&\mathbb{I}_n\\
    -\mathbb{I}_n&0\end{bmatrix},\quad \vz=[\vq,\vp]^T,
\end{equation}
where $H: \R^{2n}\to \R$ is a scalar function called Hamiltonian and $\mathbb{I}_n$ is unit symplectic matrix. 
 Suppose that we know a particular solution $t\to \vvarphi(t)\in \R^{2n}$ of the system~\eqref{eq:rs0}.  The variational equations obtained by means of linearisation of~\eqref{eq:rs0} around the solution $\vvarphi(t)$ take the form
	\begin{equation}
		\label{eq:wariacyjne}
\dot \vy= \mathbb{J}H''(\varphi(t)) \vy.
	\end{equation}
It is easy to show that the differential Galois group of this system is a subgroup of the symplectic group,
$ \mathrm{Sp}(2n,\C)$. For $n=1$ the group $\mathrm{Sp}(2,\C)$ is isomorphic to
$\mathrm{SL}(2,\C)$. However, for larger dimensions
$\mathrm{Sp}(2m,\C)\subset\mathrm{SL}(2m,\C)$, it is much smaller than $\mathrm{SL}(2m,\C)$. 
 Fortunately, the classification of subgroups of $\mathrm{Sp}(4,\C)$ has been established. 
 This classification was used in \cite{Combot:18::} where an extension of the Kovacic algorithm was developed for symplectic differential operators of degree four.

For a brief description of this algorithm, we introduce the appropriate terminology. Let $L$
be a differential operator with coefficients in $\C(z)$
\begin{equation}
  L(y)=y^{(n)}+a_{n-1}y^{(n-1)}+\cdots+a_1y'+a_0y=0,\quad a_i\in\C(z),
\end{equation}
and $A$ is its corresponding companion matrix, that is
\begin{equation}
  A=\begin{bmatrix}
    0&1&0&\cdots&0\\
    0&0&1&\cdots&0\\
    \vdots&\vdots&\vdots&\vdots&\vdots\\
    0&0&0&\cdots&1\\
    -a_0&-a_1&-a_2&\cdots&-a_{n-1}
  \end{bmatrix}.
  \label{eq:compa}
\end{equation}
The $2n$-th order operator is
\begin{itemize}
\item symplectic if there exists an invertible skew-symmetric matrix
  $W$ with elements in $\C(z)$, which satisfies
  \begin{equation}
    \label{eq:7b}
    A^TW+WA+W'=0,
  \end{equation}

\item projectively symplectic if there exists an invertible skew-symmetric matrix $W$ with  coefficients in $\C(z)$  and  $\lambda\in\C(z)$, such that
  \begin{equation}
    \label{eq:8b}
    A^TW+WA+W'+\lambda W=0.
  \end{equation}
\end{itemize}
An operator $L$ of order $n=2m$ is symplectic (respectively projectively
symplectic) when its Galois group is isomorphic to a subgroup of
symplectic matrices $\mathrm{Sp}(2m,\C)$ (respectively projectively
symplectic matrices $\mathrm{PSp}(2m,\C)$)
\[
  \begin{split}
    \mathrm{Sp}(2m,\C)=\{M\in\mathbb{M}_{2m}(\C)\ |\ M^T \mathbb{J} M=\mathbb{J}\},\quad
    \mathrm{PSp}(2m,\C)=\{M\in\mathbb{M}_{2m}(\C)\ |\ M^T \mathbb{J}
    M=\lambda \mathbb{J},\ \lambda\in\C^{\ast}\}.
  \end{split}
\]
If  $L$  is projectively symplectic, then, by multiplying it by a suitable hyperexponential function, the operator can be adjusted to become symplectic. A function  $f(z)$  is classified as hyperexponential if its logarithmic derivative,  $f'(z)/f(z)$ , is a rational function.

\begin{lemma}
  \label{lem:1}
  Assume that the system $\dot x=Ax $ is symplectic, that is, there exists an invertible skew-symmetric matrix $W$ with coefficients in $\C(z)$ which is a solution of equation~\eqref{eq:7b}. Then $x(t)$
  is a solution of $\dot x(t)=Ax(t)$ if and only if $ x_{\ast}(t)=Wx(t)$ is a solution of the adjoining equation $\dot{x}_{\ast}(t)=-A^Tx_{\ast}(t)$.
\end{lemma}
This lemma can be verified by a straightforward check. Consequently, if the symplectic operator  $L$  possesses a right factor  $L_{1}$ , then its adjoint  $L^{\ast}$  will also have a right factor  $\widetilde{L}{1}$  of the same order. Therefore, if  $L$  has a right factor  $L_{1}$ , it also possesses a left factor of the same degree.

For a system $\dot{x} = A x$, we can define its external second power system. 
It provides the following differential equation
\begin{equation}
  \label{eq:21}
  \dot W = A W -W^{T}A^{T},
\end{equation}
where $W$ is an antisymmetric matrix. Thus, Eq.~\eqref{eq:7b} represents the exterior square of the dual system associated with the system  $x{\prime} = A x$.

The classification theorem formulated in \cite{Combot:18::} is the
following.
\begin{lemma}
  \label{lem:T}
  A Lie subgroup of $\mathrm{Sp}(4,\C)$ is up to conjugacy generated
  by elements of the form:
  \begin{enumerate}
  \item upper block triangular matrices with diagonal blocks of size
    at most $2\times 2$,
  \item $2\times 2$ diagonal matrices and anti-diagonal matrices
    \[
      \begin{bmatrix}
        \ast&\ast&0&0\\
        \ast&\ast&0&0\\
        0&0&\ast&\ast\\
        0&0&\ast&\ast
      \end{bmatrix},\quad
      \begin{bmatrix}
        0&0&\ast&\ast\\
        0&0&\ast&\ast\\
        \ast&\ast&0&0\\
        \ast&\ast&0&0
      \end{bmatrix},
    \]
  \item full group $\mathrm{Sp}_{4}(\C)$.
  \end{enumerate}
\end{lemma}
This classification is constructed based on the known classification
of Lie subgroups of wider unimodular group
$\mathrm{SL}(4,\C)\supset\mathrm{Sp}(4,\C)$. Subgroups of the
projective symplectic group are central extensions of these, and so
contain multiples of the identity matrix with non-unit determinant.
But as these commute with all matrices, the possible structures of
subgroups in items $1,2$ are unchanged.

The next two lemmas characterize the reducible case.
\begin{lemma}
  \label{lem:2}
  Let us consider the following block diagonal system
  \begin{equation}
    \label{eq:2b}
    \dot x = A x \qquad 
    A= \begin{bmatrix}
      B & C\\
      0 &  D
    \end{bmatrix},
  \end{equation}  
  where $B$, $C$ and $D$ are matrices $2\times 2$ with rational coefficients. Then the equation~\eqref{eq:7b} has the following
particular solution 
\begin{equation} \label{eq:4} W = \rme^{\int r } \begin{bmatrix}
      0 & 0\\
      0 &  J
    \end{bmatrix}, \qquad J= \begin{bmatrix}
      0 & 1\\
      -1& 0
    \end{bmatrix},
  \end{equation}
  where $r$ is a rational function.
\end{lemma}

In our proof of nonintegrability, we will use the following
criterion.
\begin{lemma}
  \label{lem:5}
  Assume that the equation\begin{equation}L(y)=y^{(4)}+a_3(z)y'''+a_2(z)y''+a_1(z)y'+a_0(z)y=0,\quad '=\Dz,   
    \label{eq:4order}
  \end{equation}
  is projectively symplectic and $A$ is its corresponding companion matrix. If \eqref{eq:4order} does not admit a hyperexponential solution and Eq.~\eqref{eq:7b} has exactly one hyperexponential solution, then the differential Galois group of~\eqref{eq:4order} contains $\mathrm{Sp}(4,\C)$.
\end{lemma}
Therefore, it is essential to understand how to determine if a given equation possesses a hyperexponential solution. For an equation  $L(y) = 0$  of Fuchsian type, any hyperexponential solution must take the form  $P(z)\prod_i (z - z_i)^{e_i}$, where  $P(z) \in \mathbb{C}[z]$,  $z_i \in \mathbb{C}$  represents a singular point, and  $e_i$  are the corresponding exponents at  $z_i$. Necessary conditions for such a solution are provided in Lemma~3.1 of~\cite{Singer:95::}. From this, we derive the following proposition.
\begin{proposition}
If a Fuchsian equation of order four has a first-order factor, then either  $L(y) = 0$  or  $L^{\ast}(y) = 0$  admits a solution of the form  $P(z)\prod_i (z - z_i)^{e_i}$, where  $P(z) \in \mathbb{C}[z]$, and  $z_i \in \mathbb{C}$  are the singularities, and  $e_i$  are the exponents at  $z_i$ . Additionally, there exists an exponent at infinity,  $e_{\infty}$ , such that the total sum  $\sum_i e_i + e_{\infty}$  is a non-positive integer.
  \label{prop:reduc}
\end{proposition}
For an equation  $L(y) = 0$  that is not of Fuchsian type, determining the conditions for a hyperexponential solution is more intricate. Here, we focus on a specific scenario where the equation  $L(y) = 0$  has only a single irregular singularity at  $z = 0$. In this case, any hyperexponential solution must take the form
\begin{equation}
y(z) = \mathcal{E}(z) R(z) \prod_{i=1}^m (z - z_i)^{e_i},
\end{equation}
where  $R(z)$  is a rational function,  $z_i \in \mathbb{C}$  denotes singular points, and  $e_i$  are the corresponding exponents at  $z_i$  for  $i = 1, \dots, m$. The function  $\mathcal{E}(z)$  represents the exponential part of a formal solution at the irregular point  $z = 0$, which yields
\begin{equation}
\label{eq:exp_sol}
    \widehat{y}(z) =
\scE(z)\left[ c_0(z^{1/k})+c_1(z^{1/k})\ln(z)+\cdots + c_1(z^{1/k})\ln(z)^m\right],\quad  \scE(z)	= \exp[W(z^{-1/k})]z^a,
\end{equation}
where 
$W$ is a polynomial, $k>0$ and $m\geq 0$,  $a\in\C$,  and  $c_i$ are formal power
series.   In practice, formal solutions can be found with the help of a computer
algebra system, for example, Maple or  diagramLE.

For a differential equation of order  $n$, there are  $n$  linearly independent solutions of the form described previously. For the exponential part  $\mathcal{E}(z)$  given in equation~\eqref{eq:exp_sol}, we consider only those solutions for which  $k = 1$  and  $m = 0$. Specifically, if the exponential part is of the form  $\mathcal{E}(z) = z^a$, where  $a$  is an exponent at the singular point  $z = 0$, then the corresponding solution structure is similar to the one found in Fuchsian equations. In this scenario, if a hyperexponential solution exists, it has the same form as a Fuchsian equation.

\bibliographystyle{unsrt}

\end{document}